\documentclass[a4paper,11pt]{article}

\usepackage{amsmath,amsthm,filecontents}
\usepackage{a4wide}
\usepackage{cite}
\usepackage{url}

\usepackage{tikz}
\usetikzlibrary{decorations.markings}
\usetikzlibrary{decorations.pathreplacing}
\usepackage{color}

\newcommand\prob[3]{
  \begin{description}
    \item[\it Name.] #1
    \vspace*{-2mm}
    \item[\it Instance.] #2
    \vspace*{-2mm}
    \item[\it Output.] #3
  \end{description}
}
\newtheorem{theorem}{Theorem}
\newtheorem{lemma}[theorem]{Lemma}
\newtheorem{corollary}[theorem]{Corollary}
\newtheorem{conjecture}[theorem]{Conjecture}
\theoremstyle{definition}
\newtheorem{definition}[theorem]{Definition}
\newtheorem{observation}[theorem]{Observation}

\title {Counting $4\times 4$ Matrix Partitions of Graphs\thanks{
Final author-prepared manuscript
 of  {\tt http://dx.doi.org/10.1016/j.dam.2016.05.001}.
    The research leading to these results has received funding from
    the European Research Council under the European Union's Seventh
    Framework Programme (FP7/2007--2013) ERC grant agreement
    no.\ 334828. The paper reflects only the authors' views and not
    the views of the ERC or the European Commission. The European
    Union is not liable for any use that may be made of the
    information contained therein.}}
\author{Martin Dyer\thanks{School of Computing, University of Leeds,
    UK.} \and Leslie Ann Goldberg\thanks{Department of Computer
    Science, University of Oxford, UK.} \and David
    Richerby\footnotemark[3]}
 \date {20 June 2016}

\newcommand{\calP}{\mathcal{P}}

\newcommand{\Kkplus}{\Gamma^{1}_k}
\newcommand{\Kkbar}{\Gamma^{0}_k}
\newcommand{\badones}{\mathcal{E}}
\newcommand{\ourinst}{J}

\def\numpairs{\Delta}

\newcommand\FP{\ensuremath{\mathrm{FP}}}
\newcommand\NP{\ensuremath{\mathrm{NP}}}
\newcommand\coNP{\ensuremath{\mathrm{co}\text{-}\NP}}
\newcommand{\numP}{\ensuremath{\mathrm{\#P}}}

\def\IS{\#\text{IS}}
\def\Clique{\#\text{Clique}}
\newcommand\nListPartitions[1]{{\sc \#List-$#1$-partitions}}
\newcommand\nPartitions[1]{{\sc \#$#1$-partitions}}

\newcommand\powerset[1]{{\calP(#1)}}
\newcommand\lists{\ensuremath{\mathcal L}}
\newcommand{\mtwo}[4]{\left(\begin{smallmatrix}{#1} & {#2}\\ {#3} & {#4}\end{smallmatrix}\right)}

\DeclareRobustCommand{\stirnum}{\genfrac\{\}{0pt}{}}
\newcommand{\ffact}[2]{(#1)_{#2}}

\begin{document}
\maketitle{}

\begin{abstract}
    Given a symmetric matrix $M\in \{0,1,*\}^{D\times D}\!$, an
    $M$-partition of a graph~$G$ is a function from $V(G)$ to~$D$ such
    that no edge of~$G$ is mapped to a $0$ of~$M$ and no non-edge to
    a~$1$.  We give a computer-assisted proof that, when $|D|=4$, the
    problem of counting the $M$-partitions of an input graph is either
    in \FP{} or is \numP{}-complete.  Tractability is proved by
    reduction to the related problem of counting list $M$-partitions;
    intractability is shown using a gadget construction and
    interpolation.  We use a computer program to determine which of
    the two cases holds for all but a small number of matrices, which
    we resolve manually to establish the dichotomy.  We conjecture
    that the dichotomy also holds for $|D|>4$.  More
    specifically, we conjecture that, for any symmetric matrix
    $M\in\{0,1,*\}^{D\times D}\!$, the complexity of counting
    $M$-partitions is the same as the related problem of counting list
    $M$-partitions.
\end{abstract}

\section{Introduction}

Let $M$ be a symmetric matrix in $\{0,1,*\}^{D\times D}\!$. An
\emph{$M$-partition} of an undirected graph $G=(V,E)$ is a
partition of $V$ into parts labeled by the elements of~$D$ (some of which may be empty). The partition is represented as a
function
$\sigma\colon V\to D$
where $\sigma(v)$ is the part of vertex~$v$.
It satisfies the following property: For all pairs of distinct vertices $u$ and~$v$,
\begin{itemize}
\item $M_{\sigma(u),\sigma(v)}\in \{1,*\}$ if $(u,v)\in E$ and
\item $M_{\sigma(u),\sigma(v)}\in \{0,*\}$ if $(u,v)\not\in E$.
\end{itemize}
Thus, if $M_{i,j}=0$, no edges are permitted between vertices in parts
$i$ and~$j$ and, if $M_{i,j}=1$, then all   edges must be
present between the two parts.  If $M_{i,j}=*$, there is no
restriction on edges between parts $i$ and~$j$.
Note that self-loops play no role --- the
property  applies only to pairs of distinct vertices~$u$
and~$v$.

$M$-partitions were introduced by Feder, Hell, Klein and
Motwani~\cite{FHKM1999:Partitions, FHKM2003:Partitions} to study graph
partition problems arising in the proof of the strong perfect graph
conjecture, such as recognising skew cutsets, clique-cross partitions,
two-clique cutsets and Winkler partitions.  A skew cutset of a
connected graph~$G=(V,E)$ is a pair of disjoint, non-empty sets
$A,B\subset V$ such that $A\cup B$ is a cutset (deleting the vertices
in $A$ and~$B$ disconnects the graph) and $G$~contains every possible
edge between $A$ and~$B$.  Skew cutsets correspond to $M$-partitions
for
\begin{equation*}
        M\;=\;\bordermatrix{
                      & A & B & C & D \cr
                  A\; & * & 1 & * & * \cr
                  B\; & 1 & * & * & * \cr
                  C\; & * & * & * & 0 \cr
                  D\; & * & * & 0 & *
            }\,.
\end{equation*}
The rows (and columns) correspond to parts $A$, $B$, $C$ and~$D$,
respectively.  Consider an $M$-partition in which every part is
non-empty. $M_{A,B}=1$ so $G$~must contain every   edge between
those two parts.  The rest of the graph must be assigned to parts $C$
and~$D$ but, with no edges allowed between those parts, each of them
must be a non-empty union of components of $G-(A\cup B)$.  Therefore,
the partition corresponds to a skew cutset.  Clique-cross partitions,
two-clique cutsets and Winkler partitions also correspond to
$M$-partition problems for $4\times 4$ matrices~$M$; see
\cite{FHKM2003:Partitions} for both the definition of these problems
and the corresponding matrices.

We study the problem of counting $M$-partitions, which was introduced by Hell,
Hermann and Nevisi~\cite{HHN}.
\prob{\nPartitions{M}.}
 {A graph $G$.}
  {$Z_M(G)$, the number of $M$-partitions of $G$.}
Note that the matrix~$M$ is considered as a parameter and is not part
of the input.  For the decision problem of determining whether an
$M$-partition of some graph exists, it is conventional to require
every part to be non-empty since, otherwise, the problem is trivial
whenever there is a $*$ on the diagonal (as is the case above).
Counting, however, includes all $M$-partitions of the graph, including
those where some parts may be empty.  Hell, Hermann and Nevisi~\cite{HHN} show
that, for any $2\times 2$ or $3\times 3$ matrix $M$, the problem
\nPartitions{M} is either in \FP{} or is \numP{}-complete.  Our main
result is an extension of this dichotomy to $4\times 4$ matrices.

\begin{theorem}
\label{thm:dichotomy}
    Let $M$ be a symmetric matrix in $\{0,1,*\}^{4\times 4}\!$.
    Then
    \nPartitions{M} is either in \FP{} or is \numP{}-complete.
\end{theorem}

Thus, we completely resolve the complexity of counting $M$-partitions
for $4\times 4$ matrices, including all the examples above.

We explain the criterion that determines whether \nPartitions{M} is in
\FP{} or \numP{}-complete for a given symmetric $4\times 4$ matrix~$M$ in the next section.
Doing this requires the related concept of \emph{list
  $M$-partitions}, also due to Feder et al.~\cite{FHKM2003:Partitions}.
Here, each vertex of the input graph comes with a list of parts in
which it is allowed to be placed.  More formally, the input to the
problem is a graph $G=(V,E)$ and a function $L\colon V\to
\powerset{D}$, where $\powerset{\cdot}$ denotes the powerset.  An
$M$-partition $\sigma$ of~$G$ \emph{respects} the function~$L$ if
$\sigma(v)\in L(v)$ for all vertices $v\in V$.
The counting list $M$-partitions problem is defined as follows.
\prob{\nListPartitions{M}.}
 {A graph $G$~and a function $L\colon V(G)\to\powerset{D}$.}
  {The number of $M$-partitions of~$G$ that respect~$L$.}
The complexity of \nListPartitions{M} for all symmetric, square $\{0,1,*\}$-matrices was recently determined by
G\"obel, Goldberg, McQuillan, Richerby and Yamakami~\cite{GGMRY}:
depending on the structure of~$M$, it is either in \FP{} or
is \numP{}-complete.

The \nPartitions{M} problem without lists is the special case of
\nListPartitions{M} where $L(v)=D$ for every vertex~$v$.
Thus, there is a trivial polynomial-time Turing reduction from
\nPartitions{M} to \nListPartitions{M}. It is not known whether there is a polynomial-time Turing reduction in the other direction. As such, the dichotomy for counting
list $M$-partitions does not necessarily translate into a dichotomy
for counting $M$-partitions without lists.

$M$-partitions are also known as trigraph homomorphisms. Trigraphs are a generalisation of graphs, introduced by Chudnovsky~\cite{Ch}, which allow $*$-edges.
Thus trigraph homomorphisms are
a generalisation of the well-known graph homomorphism problem~\cite{HN2004:Homomorphisms}.
Dyer and Greenhill~\cite{DG} showed that, for any fixed graph~$H$, the problem
of counting homomorphisms from an input graph~$G$ to~$H$ is either
in \FP{} or is \numP{}-complete, depending on the structure
of~$H$. The only polynomial-time cases are those where every component
of~$H$ is either a complete graph with a self-loop on every vertex or
a complete bipartite graph with no self-loops.   The
algorithm for the polynomial-time graph homomorphism cases is easily adapted to respect
lists so, for any graph~$H$, the problems of counting homomorphisms
to~$H$ with and without lists have the same complexity~\cite{HNlist}.

We explain the criterion for the \nListPartitions{M} dichotomy
from~\cite{GGMRY} in the following section.  It is
more complex than the criterion for graph homomorphisms, and so are
the algorithms for the polynomial-time cases.  Nonetheless, for
every symmetric matrix~$M$ of size up to $4\times 4$, it is true
that \nPartitions{M} and \nListPartitions{M} have the same
complexity.  We conjecture that this holds in general.

\begin{conjecture}
\label{conj:dichotomy}
    Let $M$ be a symmetric matrix in $\{0,1,*\}^{D\times D}\!$.
    Then
    \nPartitions{M} and \nListPartitions{M} have the same complexity.
\end{conjecture}

Proving this conjecture appears considerably more difficult than
routinely extending the methods of Dyer and Greenhill~\cite{DG}, or even those
of Bulatov's far-reaching generalisation~\cite{Bul}. The difficulty arises from
the fact that some of the most powerful techniques used in proving those dichotomies
do not seem to be applicable to the $M$-partitions problem.

\subsection{The \nListPartitions{M} dichotomy}

We now describe the complexity dichotomy for the \nListPartitions{M}
problem, from~\cite{GGMRY}.  The definitions and observation in this
section are taken from that paper.

\begin{definition} For any symmetric $M\in\{0,1,*\}^{D\times D}$ and any
sets $X,Y\in\powerset{D}$, define the binary relation
\begin{equation*}
    H^M_{X,Y}=\{(i,j)\in X\times Y\mid M_{i,j}=*\}\,.
\end{equation*}
\end{definition}

The following  notion
of rectangularity was introduced by Bulatov and Dalmau~\cite{BD}.
\begin{definition}
A relation $R\subseteq D\times D'$ is \emph{rectangular} if, for all $i,j\in D$, and $i'\!,j'\in D'\!$,
\begin{equation*}
(i,i'),(i,j'),(j,i')\in R\implies (j,j')\in R\,.
\end{equation*}
\end{definition}

 \begin{definition}
Given index sets $X$ and~$Y$, a matrix $M\in\{0,1,*\}^{X\times Y}$ is
\emph{pure} if it has no $0$s or has no~$1$s.  $M$~is
\emph{$*$-rectangular} if $H^M_{X,Y}$ is rectangular.
\end{definition}

If $M$ is a pure matrix with no $1$s, then $Z_M(G)$ is the number of
homomorphisms from the graph~$G$ to the graph whose adjacency matrix
is obtained from~$M$ by changing all $*$s to~$1$s.  If $M$~is pure
with no $0$s, $Z_M(G)$ is the number of homomorphisms of the
complement of~$G$ to the graph whose adjacency matrix is obtained
from~$M$ by changing all $1$s to~$0$s and then changing all $*$s
to~$1$s.  Thus, we sometimes refer to pure matrices as
\emph{homomorphism matrices}.

\begin{definition}
For any symmetric matrix
$M\in\{0,1,*\}^{D\times D}$, a
set $\lists \subseteq \powerset{D}$
is \emph{$M$-purifying} if, for all $X,Y\in\lists$,
$M|_{X\times Y}$ is pure, where $M|_{X\times Y}$ is the submatrix
formed by restricting to rows in~$X$ and columns in~$Y$.
\end{definition}

\begin{definition}
\label{def:derect}
An \emph{\lists{}-$M$-derectangularising sequence} of length $k$ is a
sequence $D_1,\dots,D_k$  with each $D_i \in\lists$ such that:
\begin{itemize}
\item  $\{D_1,\ldots,D_k\}$ is $M$-purifying and
\item the relation $H^M_{D_1,D_2} \circ H^M_{D_2, D_3} \circ \dots \circ
H^M_{D_{k-1}, D_k}$ is not rectangular.
\end{itemize}
For brevity,
we refer to a
$\powerset{D}$-$M$-derectangularising sequence
as
an \emph{$M$-derectangularising sequence}
or as a derectangularising sequence of~$M$.
\end{definition}

\begin{observation}
\label{obs:derect-doubletons}
If there is an $i\in \{1,\ldots,k\}$ such that
$D_i=\emptyset$
then the relation $H=H^M_{D_1,D_2} \circ H^M_{D_2, D_3} \circ \dots \circ
H^M_{D_{k-1}, D_k}$
is the empty relation, which is trivially
rectangular.
If there is an~$i$ such that $|D_i|=1$ then
$H$ is a Cartesian product, and is therefore rectangular.
It follows   that $|D_i|\geq 2$ for each~$i$ in a derectangularising sequence.
\end{observation}

The complexity of \nListPartitions{M} is determined by the presence
or absence of derectangularising sequences.
The following is  \cite[Theorem 9]{GGMRY}.

\begin{theorem}
\label{thm:list}
Let $M$ be a symmetric matrix in
  $\{0,1,*\}^{D\times D}\!$.
If there is an
$M$-derectangular\-ising sequence, then
the problem \nListPartitions{M} is $\numP$-complete.  Otherwise,
 it is in $\FP$.
\end{theorem}

Thus, our conjecture that counting $M$-partitions
has the same complexity as counting list $M$-partitions is
the same as the following.

\begin{conjecture}
\label{conjecturetoquote}   \nPartitions{M} is \numP{}-complete if $M$~has a
derectangularising sequence, and is in \FP{}, otherwise.
\end{conjecture}

\subsection{Our contribution}

Our main contribution is a computer-assisted proof of
Theorem~\ref{thm:dichotomy}. This establishes a dichotomy for
\nPartitions{M} for $4\times 4$ matrices that is consistent with
Conjecture~\ref{conj:dichotomy}.  We also show that Hell, Hermann and
Nevisi's dichotomy for
$2\times 2$ and
$3\times 3$ matrices is consistent with our conjecture.

There are sufficiently few $2\times 2$ and $3\times 3$
$\{0,1,*\}$-matrices that Hell, Hermann and Nevisi were able to
determine the complexity of \nPartitions{M} for all such matrices by
case analysis.  However, this approach does not seem feasible for
larger matrices.

Recall that, for any symmetric matrix $M\in\{0,1,*\}^{D\times D}\!$,
\nPartitions{M} is the special case of \nListPartitions{M} in which
every vertex of the input graph is given list~$D$.  So, if
\nListPartitions{M} is in \FP{}, so is \nPartitions{M}.  By
Theorem~\ref{thm:list}, this occurs precisely when there is no
$M$-derectangularising sequence.  In Section~\ref{sec:easy}, we give a
method that can be used to show that some $4\times 4$ matrices do not
have $M$-derectangularising sequences.

In Section~\ref{sec:hard}, we develop gadget-based techniques for
showing \numP{}-completeness of \nPartitions{M} for symmetric $D\times
D$ matrices~$M$.  Given an input graph~$G$, we attach a
gadget~$\Gamma$ to~$G$.  The parts of~$D$ into which the vertices of the
gadget are placed determine the parts into which the vertices of~$G$
can be placed.  If we could restrict to favourable partitions of the
gadget, this would, in many cases, restrict $G$ to be partitioned
according to some proper submatrix~$M'$ for which \nPartitions{M'} is
known to be \numP{}-complete by the work of Hell et al.~\cite{HHN}.

We do not know how to restrict to specific partitions of the gadget.
However, by varying the size of the gadget and using interpolation as
follows, we are able to restrict to certain classes of partitions.
This is enough to prove hardness in all but a few cases, by showing
that we can use an oracle for \nPartitions{M} to compute
\nPartitions{M'} for some hard submatrix $M'$ of~$M$.  In more
detail, let $J(\Gamma,G)$ be the graph that results from attaching the
gadget~$\Gamma$ to the graph~$G$.  (In fact, we have two different
ways of attaching the gadget, which are described in
Section~\ref{sec:hard}; we do not need the details, here.)  For a set
$S\subseteq D$, let $Z_M^S(\Gamma)$ be the number of $M$-partitions of
the gadget~$\Gamma$ where exactly the parts in~$S$ are non-empty.  In
$M$-partitions of $J(\Gamma,G)$, placing $\Gamma$~in the parts in~$S$
restricts the vertices of~$G$ to being placed in some set
$E(S)\subseteq D$ of the parts.  We can write
\begin{equation*}
    Z_M(J(\Gamma,G)) = \sum_{S\subseteq D} Z^S_M(\Gamma)\,Z_{M|_{E(S)}}(G)\,,
\end{equation*}
where $M|_{E(S)}$ is the principal submatrix of~$M$ containing exactly
the rows and columns with indices in~$E(S)$.

The gadget~$\Gamma$ is just a clique or independent set of size~$k$ so
$Z^S_M(\Gamma)$~is a polynomial-time computable function of $M$
and~$k$.  Having computed these values, and also used the oracle to
compute $Z_M(J(\Gamma,G))$, we can view the above equation as a linear
equation in the ``variables'' $Z_{M|_{E(S)}}(G)$.  By varying the size
of the gadget, we can obtain a system of equations of this form, which
we would hope to be able to solve. However, it is usually the case
that there are distinct subsets $S_1, \dots, S_r$ of~$D$ for which the
functions $Z^{S_i}_M(\Gamma)$ for $1\leq i\leq r$ are identical.  In
this case, we cannot solve for the variables $Z_{M|_{E(S_i)}}(G)$ individually but
we can compute a weighted sum of them.  In most cases, it turns out
that only one of these variables is a \numP{}-complete function.  We
can compute the weighted sum in polynomial time from the system of
equations, and then compute all but one of the terms of that sum in
polynomial time (with the assistance of the oracle, if needed), which
allows us to compute a \numP{}-complete function, completing the
reduction from the problem of computing that function to
\nPartitions{M}.

We prove Theorem~\ref{thm:dichotomy}
with the aid of a computer program that, for each symmetric matrix
$M\in\{0,1,*\}^{4\times 4}$ attempts to use the techniques of
Section~\ref{sec:easy} to prove tractability and the interpolation
technique of Section~\ref{sec:hard} to prove intractability.  This is
described in Section~\ref{sec:computer}.  The program resolves nearly
all cases; the six exceptions (up to symmetries of the problem) are
dealt with separately in Section~\ref{sec:handproofs}.
Finally, in Section~\ref{sec:conjecture}, we show that our dichotomy
for $4\times 4$ matrices is consistent with our conjecture for the
general case, Conjecture~\ref{conjecturetoquote}.

A similar computer-assisted proof could, in principle, be applied to
$5\times 5$ matrices, the number of which is not excessive (at most
$3^{15} < 14,400,000$, even before symmetries are considered).  Doing so
requires automating more sophisticated handling of the sets of
simultaneous linear equations and seems likely to result in a larger
number of exceptional matrices than the six $4\times 4$ matrices.

\section{Preliminaries}

\paragraph{Sets.} We write $\powerset{D}$ for the powerset of~$D$ and
$D^{(k)}$ for the set of $k$-element subsets of~$D$.  For convenience,
we often list the elements of small sets as tuples (e.g., $ac$ for
$\{a,c\}$).  For any natural number~$k$, $[k]$~denotes the set $\{1,
\dots, k\}$.

\paragraph{Graphs.}
Since self-loops and parallel edges play no role in matrix partitions,
we will assume that input graphs do not have self-loops or parallel edges.
Let $\Kkplus$ be the $k$-vertex complete graph and
let $\Kkbar$ be the $k$-vertex empty graph.\footnote{This nonstandard notation
allows us to talk about a graph $\Gamma^\tau_k$ for $\tau\in\{0,1\}$,
simplifying the description of our gadget construction.}  Let $\IS(G)$ and
$\Clique(G)$ be the problems of determining, respectively, the number
of independent sets and complete subgraphs of~$G$.

\paragraph{Combinatorics.}
We write $\ffact{n}{k}$ for the falling factorial $n(n-1)\cdots(n-k+1)$,
taking $\ffact{n}{0}=1$.
\begin{equation*}
    \stirnum{n}{k} = \frac{1}{k!} \sum_{j=0}^k {(-1)}^{k-j} \binom{k}{j} {j}^{n}
\end{equation*}
denotes a Stirling number of the second kind.
The number of surjective functions from a set of size~$n$
to a set of size~$k$ is $k!\stirnum{n}{k}$.  
We will use the following bounds on $\stirnum{n}{k}$:  
\begin{equation}\label{eq:Stirling}
\mbox{For $n \geq k \ln 2k$}, \quad \tfrac12{k^n}/{k!}\leq \stirnum{n}{k}\leq {k^n}/{k!}.
\end{equation}
To see this,  consider
 \begin{equation*}
\stirnum{n}{k} = \sum_{j=0}^k {(-1)}^{k-j} \frac{j^n}{j!(k-j)!}=\sum_{j=0}^k {(-1)}^{k-j}s_j=S,
\end{equation*}
say. Now, $s_0=0$, $s_1>0$ and, for $j>1$,
\[ \frac{s_j}{s_{j-1}}=\frac{(j-1)!(k-j+1)!j^n}{j!(k-j)!(j-1)^n} =\frac{k-j+1}{j}\left(\frac{j}{j-1}\right)^n\geq\, 2,\]
if
$(1-1/j)^n\leq (k-j+1)/(2j)$. Now $(1-1/j)^n\leq e^{-n/k}$, using $1-x\leq e^{-x}$ and $1<j\leq k$. Also, $(k-j+1)/j\geq 1/k$ for $j\leq k$. Thus $s_j/s_{j-1}\geq 2$ if $e^{-n/k}\leq 1/(2k)$, i.e. $n\geq k\ln 2k$.

Thus, for $n\geq k\ln 2k$, $S$ is an alternating series with strictly increasing terms. It follows that $s_k-s_{k-1}\leq S\leq s_k$. Equation~\eqref{eq:Stirling} now follows, since $s_{k-1}\leq \frac12 s_k$ and $s_k=k^n/k!$.

\paragraph{Matrices.}
Let $M$ be a symmetric $\{0,1,*\}$-matrix with rows and columns
indexed by a finite set~$D$.  For the $4\times 4$ case, we adopt the
convention that $D=\{a,b,c,d\}$ and we index the rows (and columns)
$a$, $b$, $c$ and~$d$ from top to bottom (left to right).

For sets $S,T\subseteq D$, we write $M|_{S\times T}$ for the submatrix
of~$M$ obtained by restricting to the rows in~$S$ and the columns
in~$T$.  $M|_S$~denotes the principal submatrix $M|_{S\times S}$.

Given a symmetric $D\times D$ matrix~$M$
and another symmetric $D'\times D'$ matrix~$M'$
with $|D|=|D'|$ we   write $M\equiv M'$ if there
is a bijection $\rho\colon D\to D'$ such that $M_{i,j} =
M'_{\rho(i), \rho(j)}$ for all $i,j\in D$.  It is clear that, if
$M\equiv M'\!$, then \nPartitions{M} and \nPartitions{M'} have the same
computational complexity.

We write $\overline{M}$ for the matrix obtained from~$M$ by swapping
all $0$s and~$1$s.
Note that the $M$-partitions of any
graph~$G$ correspond directly to $\overline{M}$-partitions of the
complement of~$G$.  Write $M\approx M'$ if $M\equiv M'$ or $M\equiv
\overline{M'}$.  Again, if $M\approx M'\!$, then \nPartitions{M} and
\nPartitions{M'} have the same computational complexity.

We say that a matrix~$M$ is \emph{easy} if the problem \nPartitions{M}
is in \FP{} and \emph{hard} if it is \numP{}-complete.

\section{$2\times 2$ and $3\times 3$ matrices}
\label{sec:small}

Conjecture~\ref{conj:dichotomy} is
already known to hold for pure matrices.
As we noted earlier, in this case $Z_M(G)$ is the number of homomorphisms
from~$G$ (or its complement) to a graph whose edges correspond to the stars in~$M$.
The tractability criterion of Dyer and Greenhill~\cite[Theorem 1.1]{DG}
for graph-homomorphism counting problems
coincides with the tractability criterion for the problem with lists~\cite[Theorem 4]{HNlist}.
The condition stated in these works
concerns the graph~$H$ whose vertices are elements of~$D$ and whose
edges (including self-loops) correspond to the stars in~$M$.
The tractability condition is that each component of~$H$ is either a complete graph in which every vertex
has a self-loop or a complete bipartite graph in which no vertices have self-loops.
Bulatov and Dalmau~\cite[Theorem 12]{BD}
showed that this condition is equivalent
to the condition that the relation $H_{D,D}^M$ is rectangular,
which, in turn, is equivalent to the condition that $M$ does not
have $\mtwo***0$ or $\mtwo***1$ or any permutation of these
as a submatrix.

Conjecture~\ref{conj:dichotomy}
is also known to hold for impure $2\times 2$ matrices.
In particular,
Hell, Hermann and Nevisi~\cite[Theorem 1]{HHN}
showed that for every impure symmetric $2\times 2$ matrix~$M$,
\nListPartitions{M} is in \FP, hence so is \nPartitions{M}.

Hell, Hermann and Nevisi's dichotomy \cite[Theorem 10]{HHN}
shows that if $M$ is a symmetric impure $3\times 3$ matrix
then \nPartitions{M} is $\numP$-hard
if $M$ contains $\mtwo***0$ or $\mtwo***1$ (or any permutation of these)
as a principal submatrix. Otherwise,  \nPartitions{M} is in \FP.
We will now show that this result is consistent with Conjecture~\ref{conjecturetoquote},
which we have already shown to be equivalent to Conjecture~\ref{conj:dichotomy}.
In one direction, if $M$  contains one of these hard principal submatrices
then the rows and columns of this hard principal submatrix
are an $M$-derectangularising sequence, so Conjecture~\ref{conjecturetoquote}
also says that $M$ is hard.
In the other direction, if $M$ does not contain one of these hard principal submatrices
then the following lemma shows
that $M$ has no derectangularising sequence, so Conjecture~\ref{conjecturetoquote}
also says that $M$ is easy.

\begin{lemma}
\label{lemma:HHNrephrased}
Let $M$ be an
    impure $3\times 3$ symmetric $\{0,1,*\}$-matrix $M$ with no
    principal hard $2\times 2$ submatrix. Then $M$ has no derectangularising
    sequence.
\end{lemma}
\begin{proof}
    Let $D_1, \dots, D_k$ be a sequence of subsets of~$D = \{a,b,c\}$.
    By Observation~\ref{obs:derect-doubletons}, if $|D_i|<2$ for
    any~$i$, the sequence cannot be derectangularising; if $|D_i|=3$
    for any~$i$, the sequence is not derectangularising, since
    $M|_{D_i\times D_i} = M$ is not pure.  Thus, $|D_i|=2$ for
    all~$i$.

    {\bf Case 1.} First, suppose that $M$~has a non-principal hard $2\times 2$
    submatrix: without loss of generality, we may assume that
    $M|_{ab\times bc}$ contains three $*$s and one~$0$.  Since $M$~is
    impure, at least one of $M_{a,a}$ and~$M_{c,c}$ must be~$1$:
    without loss of generality, assume that $M_{a,a}=1$.  In fact, we
    must have $M|_{ab\times bc} = \mtwo**0*$ as, otherwise, every choice of
    $M_{c,c}$ would leave $M$~containing a hard principal $2\times 2$
    submatrix.  Therefore, $M = \left(\begin{smallmatrix} 1 & * & *
      \\ * & 0 & * \\ * & * & x \end{smallmatrix}\right)$ and
    $x\in\{0,1\}$ since otherwise $M|_{ac}$ would be hard.  The two
    choices for $x$ lead to matrices that are $\approx$-equivalent, so
    we may assume that $x=0$.

    No derectangularising sequence can include $\{a,b\}$ or $\{a,c\}$
    since $M|_{ab}$ and~$M|_{ac}$ are impure.  This leaves only
    $\{b,c\}$, but $H^M_{\{b,c\},\{b,c\}}$ is the disequality relation on
    the set $\{b,c\}$. Composing this with itself any number of times
    results in
either equality or disequality, both of which are $*$-rectangular.     Thus,
    $M$~has no derectangularising sequence.

    {\bf Case 2.} Finally, suppose that $M$~has no non-principal hard $2\times 2$ submatrix.
    Let $M'$ be the pure matrix formed from $M$ by replacing every $1$
    with a~$0$.
$M'$   does not have $\mtwo***0$   or any permutation of  this
as a submatrix. Equivalently, $H^{M'}_{D,D}$ is rectangular
and the graph whose edges correspond to stars in~$M'$
has the property that every component is a complete graph in which every vertex has a self-loop or
a complete bipartite graph in which no vertices have self-loops.
There are only three elements in~$D$
so it is easy to  see that  $M'$ has no derectangularising sequence.
Since any $M$-derectangularising sequence  is also
an $M'$-derectangularising sequence, it follows that
there is no $M$-derectangularising sequence.
\end{proof}

\section{Tractability via \nListPartitions{M}}
\label{sec:easy}

For any symmetric $D\times D$ matrix $M$, recall that \nPartitions{M} is the
special case of \nListPartitions{M} where the list of allowable parts
for every vertex is $D$.  Thus, if there is a polynomial-time
algorithm for \nListPartitions{M}, a polynomial-time algorithm for
\nPartitions{M} is immediate.

By Theorem~\ref{thm:list}, \nListPartitions{M} is in \FP{} if $M$~has
no derectangularising sequence.  Determining that a general symmetric
matrix has no derectangularising sequence is
\coNP{}-complete~\cite[Theorem~10]{GGMRY}.  However, there are only
finitely many $4\times 4$ $\{0,1,*\}$-matrices, so hardness of the
general problem is moot.
By \cite[Lemma~27]{GGMRY}, any matrix in $\{0,1,*\}^{4\times 4}$ that
has a derectangularising sequence has one of length at most 33,280 but
it is not feasible to try all such sequences.
In this section, we show that, in some cases,
it is simple to
determine that a $4\times 4$ matrix has no derectangularising
sequence.

\begin{lemma}
\label{lemma:doubletons}
    Let $M$ be a symmetric matrix in $\{0,1,*\}^{D\times D}$ such
    that, for every $W\subseteq D^{(2)}$, at least one of the
    following holds:
    \begin{enumerate}
    \item There are $S,T\in W$ (not necessarily distinct) such that
        $M|_{S\times T}$ is not pure,\label{enum:prop1}
    \item $W=\{S,T\}$, $S\cap T = \emptyset$
        and $M|_{S\times T}$ is pure and $*$-rectangular, or\label{enum:prop2}
    \item $M|_{\bigcup W}$ is pure and has no derectangularising
        sequence.\label{enum:prop3}
    \end{enumerate}
    Then \nPartitions{M} is in \FP{}.
\end{lemma}
\begin{proof}
    If $\FP=\numP$, then \nPartitions{M} is in \FP{} for any
    matrix~$M$.  So we may assume that $\FP\neq\numP$ for the rest of
    the proof.

    We prove the contrapositive.  If \nPartitions{M} is not in \FP{},
    then, by Theorem~\ref{thm:list} and the assumption that $\FP\neq\numP$,
    $M$ has a derectangularising sequence. Choose such a sequence
    $D_1, \dots, D_\ell$ that contains the least possible number of
    distinct sets among the $D_i$ (i.e., a sequence that minimises
    $|\{D_1, \dots, D_\ell\}|$).
    We show that none of the three properties holds
    for $W=\bigcup_i D_i^{(2)}\!$.  By
    Observation~\ref{obs:derect-doubletons}, $|D_i|\geq 2$ for each
    $i\in[\ell]$.

    For property~\ref{enum:prop1}, consider any $S\subseteq D_i$ and
    $T\subseteq D_j$ for any $i,j\in[\ell]$.  $M|_{D_i\times D_j}$ is
    pure because $D_1, \dots, D_\ell$ is $M$-purifying, so
    $M|_{S\times T}$ is pure.  For property~\ref{enum:prop3}, suppose
    that $M|_{\bigcup W}$ is pure, since there is nothing more to
    prove if it is not.  Since $|D_i|\geq 2$ for each~$i$, $\bigcup
    W=\bigcup_i D_i$.  Therefore, $D_1, \dots, D_\ell$ is a
    derectangularising sequence of $M|_{\bigcup W}$.

    It remains to show that $W$ does not have
    property~\ref{enum:prop2}.  Suppose that $W=\{S,T\}$ and $S\cap T
    = \emptyset$.  If there were a $D_i$ with $|D_i|>2$, we would have
    $|W|>2$, contradicting the assumption that $W=\{S,T\}$.  Thus,
    $D_i\in \{S,T\}$ for each $i\in[\ell]$.  By the definition of
    derectangularising sequence, $M|_{S\times S}$, $M|_{S\times T}$
    and~$M|_{T\times T}$ are all pure.

    $M|_{S\times S}$ and~$M|_{T\times T}$ must both be $*$-rectangular
    since, otherwise, $S,S$ or~$T,T$ would be a derectangularising
    sequence, contradicting the choice of $D_1, \dots, D_\ell$.  If
    there is some~$i\in[\ell-1]$ such that $D_i=D_{i+1}=S$, then
    $H^M_{S,S}$ must be either the equality or disequality relation
    on~$S$: any other relation would either not be rectangular or
    would prevent the sequence $D_1, \dots, D_\ell$ from being
    derectangularising.  Similarly, if we have $D_i=D_{i+1}=T$ for
    some~$i$, then $H^M_{T,T}$ must be equality or disequality on~$T$.

    There must be some $i\in[\ell-1]$ such that $D_i\neq D_{i+1}$.
    Without loss of generality, we may assume that $D_i=S$ and
    $D_{i+1}=T$.  Consider $H^M_{S,T}$.  If this were a matching or
    the complete relation $S\times T$, or if the projection onto its
    first and second columns were not $S$ and~$T$, respectively, then
    $D_1, \dots, D_\ell$ would not be derectangularising.  The only
    remaining possibility is that $H^M_{S,T}$ is not rectangular,
    i.e., $M|_{S\times T}$ is not $*$-rectangular.
\end{proof}

Given a $4\times 4$ matrix~$M$, it is easy to check whether, for each
of the $64$ subsets of~$D^{(2)}\!$, at least one of the three
properties of Lemma~\ref{lemma:doubletons} holds.  If this is the
case, we may deduce that $M$~has no derectangularising sequence so is
easy, even with lists.

\section{Identifying hard matrices}
\label{sec:hard}

For matrices $M$ that are impure and, thus, not homomorphism matrices,
we use a gadget construction and interpolation to ``pick out''
principal submatrices $M'$ for which \nPartitions{M'} is
\numP{}-complete.
While we will be concerned with $4$-element domains,
the techniques in this section could potentially also be applied to
arbitrary domains~$D$, perhaps as part of a proof of a complexity
dichotomy for all \nPartitions{M} problems, by induction on the size
of the domain.

Given a Boolean value $\tau\in\{1,0\}$,
a graph~$G$ and a positive integer~$k$,
let $\ourinst^{0,\tau}(k,G)$ be the disjoint union of~$G$ and $\Gamma^\tau_k$.
The ``$0$'' in the notation is to remind us that there are no edges between~$G$
and the ``gadget'' $\Gamma^\tau_k$ (which is a complete graph if $\tau=1$ and
a graph with no edges if instead $\tau=0$).
Also, let
$\ourinst^{1,\tau}(k,G)$ be
the graph with vertex set $V(G)\cup V(\Gamma^\tau_k)$
and edge set $E(G) \cup E(\Gamma^\tau_k) \cup (V(G) \times V(\Gamma^\tau_k))$.
The ``$1$'' in the notation is to remind us that all edges are present between~$G$
and the gadget $\Gamma^\tau_k$.

The set of $M$-partitions of $\ourinst^{\pi,\tau}(k,G)$ can be broken down according to the
set of parts $S\subseteq D$ in which vertices of the gadget~$\Gamma^\tau_k$
are placed.
For example, consider the matrix
\begin{equation}
\label{eq:example-M}
M\;=\;\bordermatrix{
  & a & b & c & d \cr
a\; & 0 & 0 & 1 & * \cr
b\; & 0 & 0 & 1 & 1 \cr
c\; & 1 & 1 & 1 & 1 \cr
d\; & * & 1 & 1 & *}
\end{equation}
and take $\pi=\tau=0$.
In an $M$-partition of~$\ourinst^{0,0}(k,G)$ in which the vertices of the $\Kkbar$ are
all in part $d$, the vertices of~$G$ must be placed in parts $a$
and~$d$.  Thus, the number of $M$-partitions of $\ourinst^{0,0}(k,G)$ in which the
$\Kkbar$ is entirely within part~$d$ is equal to the number of
$M|_{ad}$-partitions of~$G$, which is the number of independent sets
in~$G$.  If we could restrict attention to only the $M$-partitions
of~$\ourinst^{0,0}(k,G)$ in which the $\Kkbar$ is in part~$d$, we could prove
\numP{}-completeness of \nPartitions{M} by reduction from counting
independent sets which, in the guise of monotone 2-SAT, was shown to be
\numP{}-complete by Valiant~\cite{Val1979:Enumeration}.
Unfortunately, we do not know how to restrict
partitions in this way but, in this section, we set up machinery that
nonetheless allows us to develop this idea into a method for proving
hardness.

\begin{definition}
\label{defn:S-surjective}
    Let $M$ be a symmetric matrix in $\{0,1,*\}^{D\times D}$ and let
    $S\subseteq D$.  An $M$-partition~$\sigma$ of a graph~$G$ is
    \emph{$S$-surjective} if the image of~$\sigma$ is~$S$.  We write $Z^S_M(G)$ for the number of $S$-surjective
    $M$-partitions of~$G$.
\end{definition}

Given a set $S\subseteq D$, and a Boolean value $\pi\in\{0,1\}$,
let
$$E^{\pi}(S) =
\{j \in D \mid \forall i \in S, M_{i,j} \in\{\pi,*\}\, \}\,.$$
$E^{1}(S)$ is the set of parts in~$D$ that can be adjacent to every part
in~$S$; $E^{0}(S)$ is the set of parts that can be non-adjacent to every
part in~$S$.
These will be interesting to us because we will proceed as follows in our reductions.
Suppose that $M|_{E^\pi{(S)}}$ is a hard matrix and that we want to show that
$M$ is hard by reducing \nPartitions{M|_{E^\pi(S)}}
to \nPartitions{M}.
Then we can take an instance $G$ of \nPartitions{M|_{E^\pi(S)}}
and form the gadget $\ourinst^{\pi,\tau}(k,G)$ for some value of~$k$.
Then, if we can choose~$\tau$ so that the gadget $\Gamma_k^\tau$
is always partitioned surjectively into parts in~$S$,
we will have reduced \nPartitions{M|_{E^\pi(S)}} to
\nPartitions{M}.
Typically, we cannot do this, but we will be able to do is to
compute the number of $M$-partitions of $\ourinst^{\pi,\tau}(k,G)$ for lots of values of~$k$.
Using polynomial interpolation, we will be able to work out
the number of $M$-partitions of~$G$ which are consistent with an $S$-surjective partition
of $\Gamma_k^\tau$ so this will enable us to
count the $M|_{E^\pi(S)}$-partitions of~$G$ (solving a hard problem) by using
an oracle for counting $M$-partitions.
Thus, we will have proved that $M$ is a hard matrix.

For $\pi \in \{0,1\}$, we say that a principal submatrix
$M'$ of~$M$ is \emph{$(M,\pi)$-accessed by~$S$}
  if
$M'\equiv M|_{E^\pi(S)}$.
Note the equivalence ---
$M'$ only has to be equivalent to
$M|_{E^\pi(S)}$ --- it doesn't have to \emph{be}
$M|_{E^\pi(S)}$. It is useful to define things this way because equivalent matrices
correspond to matrix partition problems
of equivalent difficulty. Also, we will not be able to separate them by interpolation, so we will have to consider them together.

To illustrate these definitions, consider the matrix $M$ in Equation~\eqref{eq:example-M}.
Then $E^1(\{b,d\})=\{c,d\}$
and $E^{0}(\{b,d\})= \{a\}$.
Thus,
$M|_{cd}$ is $(M,1)$-accessed by~$\{b,d\}$
and $M|_{a}$
 is $(M,0)$-accessed by~$\{b,d\}$.
 The matrix $M|_{b}$ is also $(M,0)$-accessed by~$\{b,d\}$ since
 $M|_{b} \equiv M|_{a}$.
 Also, $E^1(\{d\})=\{a,b,c,d\}$.
 Thus, $M$ itself is $(M,1)$-accessed by~$\{d\}$.

We say that
a principal submatrix
$M'$ of~$M$ is \emph{accessible} in the graph~$\ourinst^{\pi,\tau}(k, G)$
if there is a set $S\subseteq D$
such that $Z_M ^S(\Gamma^\tau_k)>0$
and $M'$ is $(M,\pi)$-accessed by~$S$.

Continuing our example with $S=\{b,d\}$ and $M$ as in Equation~\eqref{eq:example-M}, note that
for any  $k>1$,
$Z_M^S(\Kkplus)>0$ since an $S$-surjective $M$-partition of~$\Kkplus$ may
place one vertex in part~$b$ and the remaining vertices in part~$d$.
Thus,
$M|_{cd}$ is  accessible
in $\ourinst^{1,1}(k,G)$
and $M|_{a}$ and $M|_{b}$
 are  accessible in $\ourinst^{0,1}(k,G)$.
Note that accessibility in $\ourinst^{\pi,\tau}(k,G)$
depends on~$M$,  $\pi$, $\tau$  and possibly~$k$  but
it does
not depend on~$G$. Because of
this, we may talk about accessibility in $\ourinst^{\pi,\tau}(k,\cdot)$.
In fact, we will see later
in Theorem~\ref{thm:clique-partition}
that accessibility
will not actually depend on~$k$, provided that $k>|D|$ (this is not obvious at this point but will be important).

We now begin to decompose $Z_M(\ourinst^{\pi,\tau}(k,G))$
into more
manageable units.  The first step
is to break the sum up over the set~$S$ which is used to surjectively
partition the gadget $\Gamma_k^\tau$:
\begin{equation}\label{eq:one}
    Z_M(\ourinst^{\pi,\tau}(k,G)) =
    \sum_{S \subseteq D} Z^S_M(\Gamma^\tau_k)\, Z_{M|_{E^{\pi}(S)}}(G)\,.
\end{equation}
Now let
$$\Psi_\pi = \{ S  \subseteq D \mid
\text{$M$ itself is $(M,\pi)$-accessed by~$S$}\}\,.$$
The set $\Psi_\pi$ may be empty, depending on~$M$.
The reason that we have defined $\Psi_\pi$ is that wish to use Equation~\eqref{eq:one}
to show that $M$ is a hard matrix --- so we will use an oracle for $M$-parititons
to compute the left-hand side and we will hope to discover the solution to some
hard problem on the right-hand side.
For this reason we don't want $M$ itself to be one of the matrices
$M|_{E^\pi(S)}$
appearing on the right-hand side.
To ease the notation, let $\overline{\Psi_\pi}=\powerset{D}\setminus \Psi_\pi$;
$\overline{\Psi_\pi}$
consists of all subsets $S$ of~$D$ apart from those with $M|_{E^\pi(S)}=M$.
From~\eqref{eq:one}, we have
\begin{equation}\label{eq:oneone}
    Z_M(\ourinst^{\pi,\tau}(k,G))\; -    \sum_{S \in \Psi_\pi} Z^S_M(\Gamma^\tau_k)\,
    Z_M(G)\;=
    \sum_{S \in \overline{\Psi_\pi}} Z^S_M(\Gamma^\tau_k)\, Z_{M|_{E^{\pi}(S)}}(G)\,.
\end{equation}

Now we would like to collect the terms on the right-hand side of Equation \eqref{eq:oneone},
gathering all terms with the same matrix~$M|_{E^{\pi}(S)}$, and taking these together.
So, for any principal submatrix
$M'$ of $M$,
let
\begin{equation*}
    C_{M'}^{\pi,\tau}(k) = \sum_S Z_M^S(\Gamma^\tau_k)\,,
\end{equation*}
Where the sum is over sets $S \subseteq D$ such that
$M'$ is $(M,\pi)$-accessed by $S$.
Thus, $M'$~is accessible in $\ourinst^{\pi,\tau}(k,\cdot)$
precisely when $C_{M'}^{\pi,\tau}(k)$ is positive.
The quantity $C_{M'}^{\pi,\tau}(k)$ corresponds roughly to the coefficient
of $Z_{M'}(G)$ in~\eqref{eq:oneone} though we will have to be
careful about over-counting.
As a first step, we can immediately rewrite the left-hand side of \eqref{eq:oneone},
combining the terms for all $S\in \Psi_\pi$, since these terms have
a common factor of $Z_M(G)$.
\begin{equation}\label{eq:oneoneone}
    Z_M(\ourinst^{\pi,\tau}(k,G)) -   C_M^{\pi,\tau}(k)\, Z_M(G)\; =
    \sum_{S \in \overline{\Psi_\pi}} Z^S_M(\Gamma_k)\, Z_{M|_{E^{\pi}(S)}}(G)\,.
\end{equation}
Now, all of the matrices $M|_{E^{\pi}(S)}$ such that
$Z_{M|_{E^{\pi}(S)}}(G)$ arises on the right-hand side of \eqref{eq:oneoneone}
are proper principal sub-matrices of~$M$.
Since a proper principal sub-matrix
$M'$  is {$(M,\pi)$-accessed by~$S$}
  when
$M'\equiv M|_{E^\pi(S)}$,
the coefficient $C_{M'}^{\pi}$ captures the contribution of the entire equivalence class. Thus,
 we have
\begin{equation}\label{eq:two}
Z_M(\ourinst^{\pi,\tau}(k,G))-
C_M^{\pi,\tau}(k)\, Z_M(G)
= \sum_{M'} C_{M'}^{\pi,\tau}(k)\, Z_{M'}(G)\,,
\end{equation}
where the sum is over one element from each $\equiv$-equivalence class
of proper
principal sub-matrices $M'$ of~$M$.

We now explain the point of Equation~\eqref{eq:two}.
Corollary~\ref{cor:used}
will show that all of the coefficients $C_{M'}^{\pi,\tau}(k)$ can be computed in
polynomial time (as a function of~$k$).
Also, the left side of~\eqref{eq:two} can be computed
in polynomial time with an oracle for
computing~$Z_M$ --- we just use the oracle twice to compute
$Z_M(\ourinst^{\pi,\tau}(k,G))$
and $Z_M(G)$.
So if we can show that it is hard to compute the right side of~\eqref{eq:two},
then
we can conclude that computing~$Z_M$ is hard.

Since each $M'$ is a proper principal submatrix of~$M$, the
complexity of computing each $Z_{M'}$ is known from the dichotomy
of Hell, Hermann and Nevisi~\cite{HHN} and is either in \FP{} or is
\numP{}-complete.

We begin with two straightforward cases in Lemmas~\ref{lem:first} and~\ref{lemma:IS-clique}. These
cases do not require interpolation, but
we will handle these cases first and then explain the interpolation.

\begin{lemma}\label{lem:first}
Suppose that $M$ is a symmetric matrix in $\{0,1,*\}^{D\times D}$, that
$\pi$ and $\tau$ are Boolean values in $\{0,1\}$, and that $k$ is some positive integer.
If   there is at least one
proper hard submatrix of~$M$ that is  accessible in
$\ourinst^{\pi,\tau}(k,\cdot)$
and all such proper hard submatrices are  $\equiv$-equivalent,
then $M$ is hard.
\end{lemma}
\begin{proof}
    Suppose that, up to $\equiv$-equivalence, $M''$ is the only hard
    proper submatrix that is accessible in $\ourinst^{\pi,\tau}(k,\cdot)$.
    Rearranging \eqref{eq:two}, we obtain, for any graph~$G$,
    \begin{equation*}
        Z_{M''}(G) = \frac{1}{C^{\pi,\tau}_{M''}(k)}
                    \left(
                        Z_M(\ourinst^{\pi,\tau}(k,G)) - C_M^{\pi,\tau}(k)\, Z_M(G)
                         \;\; - \!\!\sum_{M'\not\equiv M''}\!\! C_{M'}^{\pi,\tau}(k)\, Z_{M'}(G)
                    \right).
    \end{equation*}
    Since all the quantities $C_{M'}^{\pi,\tau}(k)$ and $Z_{M'}(G)$
    are computable in \FP\
(which follows since all $M'$ are easy by assumption, and  the coefficients $C_{M'}^{\pi,\tau}(k)$ are constants)
    this gives a polynomial-time Turing
    reduction from \nPartitions{M''} to \nPartitions{M}.
\end{proof}

\begin{lemma}
\label{lemma:IS-clique}
Suppose that $M$ is a symmetric matrix in $\{0,1,*\}^{D\times D}$, that
$\pi$ and $\tau$ are Boolean values in $\{0,1\}$, and that $k$ is some positive integer.
Suppose that there is at least one
proper hard submatrix of~$M$ that is accessible in
$\ourinst^{\pi,\tau}(k,\cdot)$ that is $\equiv$-equivalent
to $M_0 = \mtwo***0$
and that there is at least one
proper hard submatrix of~$M$ that is accessible in
$\ourinst^{\pi,\tau}(k,\cdot)$ that is $\equiv$-equivalent
to $M_1 = \mtwo***1$.
Suppose that every  proper hard submatrix that is accessible is
either $\equiv$-equivalent
    to $\mtwo***0$ or to $\mtwo***1$.
 Then $M$~is hard.
\end{lemma}
\begin{proof}
    Recall that $\IS(G)$ and $\Clique(G)$ are, respectively, the
    number of independent sets and complete subgraphs in a
    graph~$G$. Computing each of these is \numP{}-complete~\cite{Val1979:Enumeration} and they
    correspond to \nPartitions{\mtwo***0} and \nPartitions{\mtwo***1},
    respectively.

    We first show that, for any fixed integers $\alpha$ and~$\beta$,
    computing the function $\theta_{\alpha,\beta}(G) = \alpha\IS(G) +
    \beta\Clique(G)$ is also \numP{}-complete unless $\alpha=\beta=0$.
    Assume that $\alpha$ and~$\beta$ are both non-zero as the result is
    trivial, otherwise.  Observe that, for any graph~$G$,
    \begin{gather*}
        \IS(G+K_1) = 2\IS(G) \\
        \Clique(G+K_1) = \Clique(G) + 1\,.
    \end{gather*}
    Therefore,
    \begin{equation*}
        \theta_{\alpha,\beta}(G+K_1) - \theta_{\alpha,\beta}(G)
            = \alpha\IS(G) + \beta\,,
    \end{equation*}
    which is \numP{}-complete to compute since $\alpha\neq 0$.
Thus, we have shown that computing $\theta_{\alpha,\beta}(\cdot)$ is \numP-complete.

Now rearrange \eqref{eq:two} as in the proof of Lemma~\ref{lem:first}.
 \begin{align*}
\hspace{5em}&\hspace{-5em}
 C_{M_0}^{\pi,\tau}(G)\, Z_{M_0}(G)\;+\;
 C_{M_1}^{\pi,\tau}(G)\, Z_{M_1}(G)  \\
 &=\; Z_M(\ourinst^{\pi,\tau}(k,G)) \;-\; C_M^{\pi,\tau}(k)\, Z_M(G)
    \;\;- \hspace{-1.25em}\sum_{M'\not\in \{M_0,M_1\}}\hspace{-1em} C_{M'}^{\pi,\tau}(k)\, Z_{M'}(G)\,,
\end{align*}
where the sum is over one element from each $\equiv$-equivalence class
of proper principal submatrices~$M'$ of~$M$ other than
the equivalence classes of~$M_0$ and~$M_1$.
Writing $\IS(G)$ for $Z_{M_0}(G)$ and
$\Clique(G)$ for $Z_{M_1}(G)$, and taking
$\alpha = C_{M_0}^{\pi,\tau}(k)$
and $\beta = C_{M_1}^{\pi,\tau}(k)$, we get
\begin{equation*}
\theta_{\alpha,\beta}(G)
 \;=\; Z_M(\ourinst^{\pi,\tau}(k,G)) \;-\; C_M^{\pi,\tau}(k)\, Z_M(G)
       \;\;- \hspace{-1.25em}\sum_{M'\not\in \{M_0,M_1\}}\hspace{-1em} C_{M'}^{\pi,\tau}(k)\, Z_{M'}(G)\,.
\end{equation*}

Thus, we have reduced the
\numP-hard problem of computing $\theta_{\alpha,\beta}(\cdot)$ to
the problem of evaluating the right-hand side, which can be done in polynomial time with
an oracle for \nPartitions{M}.
We conclude that \nPartitions{M} is
    \numP{}-complete.
\end{proof}

Lemmas~\ref{lem:first} and~\ref{lemma:IS-clique}
give us a tool for identifying some hard matrices~$M$.
However, neither of these lemmas helps with our example
matrix~\eqref{eq:example-M}.  To make progress, we will use interpolation.
First, in Theorem~\ref{thm:clique-partition}, we will show that the value of $Z_M^S(\Gamma^\tau_k)$ is very constrained
--- there are only a few possible values, depending on~$k$.  Further,
in Lemma~\ref{lemma:lin-ind} we will show that
these values are linearly independent as functions of~$k$.  We will later
use this fact to prove hardness by interpolation.

\begin{definition}
    Let $f_{\ell,s}(k) = \ffact{k}{\ell}\; (s-\ell)!\, \stirnum{k-\ell}{s-\ell}$.
\end{definition}

$f_{\ell,s}(k)$ is the number of ways that a set of size~$k$ can be
partitioned into $s$~parts, the first~$\ell$ of which have size
exactly~$1$ and the remaining $s-\ell$ of which have size at
least~$1$.

\begin{definition}
Let $M$ be any symmetric matrix in $\{0,1,*\}^{D\times D}$.
Let $\tau \in \{0,1\}$ be a Boolean value.
For  $S\subseteq D$,
let $\ell(M,S,\tau) = |\{ i \in S \mid M_{i,i}=\tau \oplus 1 \}|$.
Let $$\badones(M,\tau) = \{S \mid
\text{ $\ell(M,S,\tau)=|S|$ or
there are distinct $i,j\in S$ with $M_{i,j}=\tau \oplus 1$}\}\,.$$
\end{definition}

Intuitively, $\badones(M,\tau)$ is
the set of subsets $S$ of~$D$
that will not be useful for
$S$-surjectively partitioning the gadget $K_k^\tau$
(as long as $k>|D|$).
For example, if there are distinct $i,j\in S$ with $M_{i,j}=\tau \oplus 1$
then we can't simultaneously use parts~$i$ and~$j$,
so an $S$-surjective partition is impossible. We will
see below that an $S$-surjective partition is also impossible
if $\ell(M,S,\tau)=|S|$.
The following theorem shows that as long as $S\not\in\badones(M,\tau)$
the number of $S$-surjective $M$-partitions of
$\Gamma_k^\tau$ is a simple function of~$k$.

\begin{theorem}
\label{thm:clique-partition}
Let $M$ be any symmetric matrix in $\{0,1,*\}^{D\times D}$
and suppose $S\subseteq D$ and $\tau \in \{0,1\}$.
If $S \in \badones(M,\tau)$ then, for all $k>|D|$, $Z^S_M(\Gamma_k^\tau)=0$.
Otherwise,   for all $k>|D|$,
$Z^S_M(\Gamma_k^\tau) = f_{\ell(M,S,\tau),|S|}(k)$.
 \end{theorem}

\begin{proof}

{\bf Case 1.} Suppose
 there are distinct $i,j\in S$ with $M_{i,j}=\tau \oplus 1$. Then no $M$-partition of
any $\Gamma_k^\tau$
can place elements in
 both parts $i$ and~$j$.
 Thus, for any $k$,
   there are no $S$-surjective
    $M$-partitions of $\Gamma_k^\tau$, so $Z_M^S(\Gamma_k^\tau)=0$.

 {\bf Case 2.} Suppose we are not in Case~1.
 Let $S' = \{i\in S\mid M_{i,i}=\tau \oplus 1\}$
 so $\ell(M,S,\tau) = |S'|$.
In any
$S$-surjective $M$-partition of any~$\Gamma_k^\tau$, every part in $S'$ must
contain exactly one vertex.

{\bf Case 2a.} If $S\in \badones(M,\tau)$ then
 $|S'|=|S|$, so for all $k>|D|\geq|S'|$, we have $Z_M^S(\Gamma_k^\tau)=0$.

{\bf Case 2b.} Otherwise, $S\notin \badones(M,\tau)$. Let
$\ell=\ell(M,S,\tau)<|S|$.
Now, for any $k>|S|$,
$Z^S_M(\Gamma_k^\tau) = f_{\ell,|S|}(k)$. To see this, note that
    there are $\ffact{k}{\ell}$ ways to choose one vertex of~$\Gamma_k^\tau$ to place
    in each part in~$S'$.  This leaves the remaining $k-\ell$
    vertices to be surjectively placed in the $|S|-\ell$ parts in
    $S\setminus S'$.  There are $(|S|-\ell)!\,
    \stirnum{k-\ell}{|S|-\ell}$ ways of doing this. \end{proof}

Since $f_{\ell,s}(k)$ can be evaluated in
polynomial time (as a function of~$k$), we obtain the following corollary.

\begin{corollary}\label{cor:used}
    For any symmetric matrix $M$ in $\{0,1,*\}^{D\times D}$
and any $S\subseteq D$, the
$S$-surjective
    $M$-partitions of complete and empty graphs can be counted in
    polynomial time.
\end{corollary}

\begin{lemma}
\label{lemma:lin-ind}
Suppose $|D|\geq 2$.
Then there is a full rank matrix $F$
satisfying the following properties.
\begin{itemize}
\item The columns of $F$
are
indexed by the pairs $(\ell,s)$ with $0 \leq \ell < s \leq |D|$.
\item The rows of $F$
are indexed by $\binom{|D|+1}{2}$
distinct values $k_1 < k_2 < \ldots$, all of which are greater than $|D|$.
\item For each row~$k_i$ and each column~$(\ell,s)$,
the corresponding entry in~$F$
is $f_{\ell,s}(k_i)$.
\end{itemize}
\end{lemma}
\begin{proof}

Let $d=|D|$
and let $U = \{(\ell,m) \mid
\mbox{$0\leq \ell< d $ and $1\leq m\leq d-\ell$}\}$.
For $(\ell,m)\in U$,
let $\phi_{\ell,m}(k) = f_{\ell,\ell+m}(k)$.
The stated properties of the matrix $F$
indicate that
the function~$\phi_{\ell,m}$ maps every row index~$k$   to the entry in
row~$k$ and column $(\ell,\ell+m)$ of~$F$.
Let $\Phi = \{ \phi_{\ell,m} \mid  (\ell,m) \in U\}$.

We will show that the functions in $\Phi$ (which correspond to the columns of~$F$)
are linearly independent (as functions of~$k$).
To do this, we define a strict ordering $<$ on functions in~$\Phi$.
Then we will show that for any $\phi\in \Phi$, the function~$\phi$ cannot be expressed
as a linear combination of the functions in $\{ \phi' \in \Phi \mid \phi' < \phi\}$,
because it grows too fast as $k$ increases.
Then we will also be able to conclude that
$\binom{d+1}{2}$ row indices can be chosen so that the matrix $F$ has full rank,
and the other properties in the statement of the lemma are satisfied.

We first define the ordering on the $\binom{d+1}{2}$ functions in~$\Phi$.
We do this by defining a lexicographic ordering on the set $U$ of column indices,
and then ordering the functions in~$\Phi$ accordingly.
For $(\ell',m')$ and $(\ell,m)$ in~$U$,
we say that $(\ell',m') < (\ell,m)$  if one of the following is true:
\begin{itemize}
\item $m'<m$, or
\item $m'=m$ and
$\ell'<\ell$.
\end{itemize}
We use the natural induced order on functions:
$\phi_{\ell',m'} < \phi_{\ell,m}$ if and only if
$(\ell',m') < (\ell,m)$.

For convenience, let
$\Phi_{\ell,m} = \{ \phi \in \Phi \mid \phi < \phi_{\ell,m} \}$.
We will show that $\phi_{\ell,m}$ is not in the
span of $\Phi_{\ell,m}$, for all  $(\ell,m)\in U$.
We start by deriving bounds on $\phi_{\ell,m}(k)$.
If $k$ is an integer that is  at least~$\ell+m\ln 2m$,
then, from Equation~\eqref{eq:Stirling}, we have
 \begin{equation*}
        \phi_{\ell,m}(k)
            =\ffact{k}{\ell}\; m!\, \stirnum{k-\ell}{m}
            =\left\{
               \begin{array}{l@{\ }l}
                 \leq  \ffact{k}{\ell}\; m!\,
            m^{k-\ell}/m! &= \ffact{k}{\ell}\; m^{k-\ell},\\[1ex]
                  \geq  \ffact{k}{\ell}\; m!\,
            \tfrac12 m^{k-\ell}/m!
            &= \tfrac12\ffact{k}{\ell}\; m^{k-\ell}.
               \end{array}
             \right.
    \end{equation*}
Now $k^\ell\geq \ffact{k}{\ell} \geq {(k-\ell)^\ell} = k^\ell {(1-\tfrac{\ell}{k})}^\ell \geq k^\ell(1-\ell^2/k)\geq \tfrac12 k^\ell$ if $k\geq 2\ell^2$. So, if $k\geq 2\ell^2+m\ln 2m$, then
\begin{equation}\label{eq:labelthemall}
        \phi_{\ell,m}(k)
            =\ffact{k}{\ell}\; m!\, \stirnum{k-\ell}{m}
            =\left\{
               \begin{array}{l@{\ }l}
                 \leq  \ffact{k}{\ell}\; m^{k-\ell} &\leq k^{\ell}\; m^{k-\ell},\\[0.5ex]
                  \geq \tfrac12\ffact{k}{\ell}\; m^{k-\ell}
            &\geq \tfrac14 k^{\ell}\; m^{k-\ell}.
               \end{array}
             \right.
    \end{equation}
Now, we wish to show that $\phi_{\ell,m}$ is not in the
span of $\Phi_{\ell,m}$.
The claim is trivial if $\ell=0$ and $m=1$
  since $\Phi_{0,1}=\emptyset$,
so suppose otherwise.
Consider any function~$\psi$ in the linear span
of~$\Phi_{\ell,m}$.
We will show that $\psi$ is not
equal to $\phi_{\ell,m}$.
Clearly, we can assume that $\psi$ is not identically~$0$
since $\phi_{\ell,m}$ is not identically zero.
By the definition of linear span, there are  real numbers $\beta_{\phi}$,
not depending on $k$,
so that
$    \psi(k) =
               \sum_{\phi \in \Phi_{\ell,m}}
                   \beta_{ \phi} \phi(k)$.
First suppose $m'\leq m-1$ for all $\phi_{\ell',m'}\in \Phi$. Plugging in~\eqref{eq:labelthemall}, we 
will show that, if $k$ is sufficiently large, then
\begin{equation}\label{eq:blahblah}
        \psi(k) \leq
               \sum_{\phi \in \Phi_{\ell,m}} \beta_{\phi} k^d (m-1)^k \leq \beta_{\Phi}k^d (m-1)^k
< \tfrac18 k^\ell m^{k-\ell} \leq \tfrac12\phi_{\ell,m}(k),
\end{equation}
where  $ \beta_{\Phi} = \sum_{\phi \in \Phi_{\ell,m}} |\beta_{\phi}|>0$. 
Note that $\beta_{\Phi}$ depends on $\psi$, $\ell$ and $m$ but not on~$k$. Now \eqref{eq:blahblah} holds if 
$k\geq 2\ell^2+m\ln 2m$ (for the final inequality) and
$8\beta_{\Phi} m^\ell k^d (1-1/m)^k < 1$ (for the strict inequality).
The latter inequality is true if $k^d e^{-k/m} < 1/(8\beta_{\Phi} m^\ell)$. 
Now $k^d\leq e^{k/2m}$ if 
$k/(\ln k) > 2d m$, which is true if $k > 4 m^2 d^2$ (since $\ln k < \sqrt{k}$ for all $k\geq 1$).
 So, if 
 $k >  \max(2\ell^2+m\ln 2m, 4 m^2 d^2)$,
 the condition becomes  $e^{k/2m}> 8\beta_{\Phi} m^\ell$, i.e. $k>2m\ln(8\beta_{\Phi}m^\ell)$. So, if 
 $k>\max( 2\ell^2+m\ln 2m, 4 m^2 d^2, 2m\ln(8\beta_{\Phi}m^\ell))$,  then $\psi(k) < \frac12\phi_{\ell,m}(k)$, so $\psi \neq \phi_{\ell,m}$.

In the general case, let $\Phi' = \{ \phi \in \Phi_{\ell,m} \mid \phi < \phi_{d,m-1} \}$ and $\Phi'' = \{ \phi_{\ell',m} \in \Phi_{\ell,m} \mid \ell'<\ell \}$. Thus
\[ \psi(k) = \sum_{\phi \in \Phi_{\ell,m}} \beta_{\phi} \phi(k) = \sum_{\phi \in \Phi'} \beta_{\phi} \phi(k) +\sum_{\phi \in \Phi''} \beta_{\phi} \phi(k).\]
Now, using the proof of~\eqref{eq:blahblah} above, 
\[ \sum_{\phi \in \Phi'} \beta_{\phi} \phi(k)\ <\ \tfrac12\phi_{\ell,m}(k),\]
if $k>\max(2\ell^2+m\ln 2m, 4 m^2 d^2, 2m\ln(8\beta_{\Phi'}m^\ell))$, where  $ \beta_{\Phi'} = \sum_{\phi \in \Phi'} |\beta_{\phi}|$. Also, using~\eqref{eq:labelthemall} again,
\[ \sum_{\phi \in \Phi''} \beta_{\phi} \phi(k)< \beta_{\Phi''}k^{\ell-1}m^{k-\ell}
<\tfrac18 k^\ell m^{k-\ell} \leq \tfrac12\phi_{\ell,m}(k),\]
provided that we also have $k>8 \beta_{\Phi''}$, where  $ \beta_{\Phi''} = \sum_{\phi \in \Phi''} |\beta_{\phi}|$. Thus if
\begin{equation}\label{eq:after}
k > k'=\max(2\ell^2+m\ln 2m, 4 m^2 d^2, 2m\ln(8\beta_{\Phi'}m^\ell)), 8 \beta_{\Phi''}),
\end{equation}
we have $\psi(k) < \phi_{\ell,m}(k)$, and so $\psi \neq \phi_{\ell,m}$.

Now we will show how to choose $\binom{d+1}{2}$ row indices $k_1,k_2,\ldots$,
so that $F$ has full rank, and the other properties in the statement of the lemma are satisfied. Order the columns of $F$ according to the ordering $<$ defined above.
We will choose the row-indices $k_1,k_2$ inductively,
using the invariant that $F^i$, which the sub-matrix defined by
the row-indices $k_1,\ldots,k_i$ and the first $i$ columns in~$U$,
has full rank.
The base case, $i=1$, is trivial --- for concreteness, take $k_1=d+1$.
Now consider the inductive step, and the choice of $k_{i+1}$.
Let $(\ell,m)$ denote the $(i+1)$st pair in $U$.
Since $F^i$ has full rank,
there is exactly one linear combination of the first $i$ columns of $F^i$
that agrees with the $(i+1)$st column on the rows with indices $k_1,\ldots,k_i$.
Thus, there is only one possible linear combination
$\psi$ in the linear span of $\Phi_{\ell,m}$ that that agrees with $\phi_{\ell,m}$ on $k_1,\ldots,k_i$.
Now, use \eqref{eq:after} to choose $k'$ so that
$\phi_{\ell,m}(k) > \psi(k)$ for $k> k'$, 
and set $k_{i+1}=\min(k_i,\lceil k'\rceil)+1$. This completes the inductive step, and the proof.
\end{proof}

At this point is helpful to recall our construction
of the graph $\ourinst^{\pi,\tau}(k,G)$ from~$G$. It also
 helps to recall Equation~\eqref{eq:one}.
\begin{equation*}
   Z_M(\ourinst^{\pi,\tau}(k,G)) =
    \sum_{S \subseteq D} Z^S_M(\Gamma^\tau_k)\, Z_{M|_{E^{\pi}(S)}}(G)\,.
\end{equation*}
We know from Theorem~\ref{thm:clique-partition}
that, for any matrix $M|_{E^{\pi}(S)}$ corresponding to an element~$S$ of the sum,
either
$S\in \badones(M,\tau)$ in which case
 the function $Z^S_M(\Gamma^\tau_k) $  is identically zero
(assuming $k>|D|$)
or $S\notin\badones(M,\tau)$ in which case
it is identically the function
 $f_{\ell(M,S,\tau),|S|}(k)$
(as a function of~$k$).
Let
$$ \mathcal{S}(\ell,s,M,\tau) =
\{ S \in \powerset{D} \setminus \badones(M,\tau) \text{ such that }
|S|=s \text{ and } \ell(M,s,\tau)=\ell \}\,.$$
$\mathcal{S}(\ell,s,M,\tau)$ is the set of sets $S\subseteq D$
such that $Z^S_M(\Gamma^\tau_k) = f_{\ell,s}(k)$.
Thus, we can rewrite Equation~\eqref{eq:one} for $k>|D|$
as
\begin{equation}
\label{eq:anotherone}
 Z_M(\ourinst^{\pi,\tau}(k,G)) \ \ =\!\!\!
 \sum_{0 \leq \ell < s \leq |D|}\!\!\!
 f_{\ell,s}(k)
 \sum_{S \in \mathcal{S}(\ell,s,M,\tau)} \!\!\!  Z_{M|_{E^{\pi}(S)}}(G)\,.
\end{equation}

Now the point is that the $f_{\ell,s}(k)$ entries are linearly independent
functions of~$k$ by Lemma~\ref{lemma:lin-ind}.
We will see in the proof of Theorem~\ref{thm:interpolation}
that we will be be able to choose  sufficiently many values of~$k$,
evaluate the left-hand side $Z_M(\ourinst^{\pi,\tau}(k,G))$ for each of these using an oracle for \nPartitions{M}
and then interpolate to compute each ``coefficient'' of $f_{\ell,s}(k)$ on the right-hand side.
That is, we show how to compute each  value
$\sum_{S\in\mathcal{S}(\ell,s,M,\tau)}   Z_{M|_{E^{\pi}(S)}}(G)$.
If computing one of these values (for an input~$G$) is a hard problem, then we will have proved that \nPartitions{M} is
also \numP-complete.

Before we proceed it will help to rewrite~\eqref{eq:anotherone} one last time,
splitting the sum   over principal submatrices of~$M$.
 For $0 \leq \ell < s \leq |D|$,
 let
$$A^{\pi,\tau}_M (\ell,s) =
   \{M|_{E^\pi(S)} \mid S \in \mathcal{S}(\ell,s,M,\tau)\}\,.$$
$A^{\pi,\tau}_M(\ell,s)$ is
just the set of matrices $M'$ such that the coefficient of $f_{\ell,s}(k)$
in \eqref{eq:anotherone} has a $Z_{M'}(G)$ term.
As before, we will need to deal with equivalences between matrices.
Let $A^{\pi,\tau}_M(\ell,s) /{\equiv}$ be
the set containing one matrix from each $\equiv$-equivalence class
of
$A^{\pi,\tau}_M(\ell,s)$.
For each matrix $M'$ in
$A^{\pi,\tau}_M(\ell,s) /{\equiv}$, let
$$n_{M'}(\ell,s) =   |\{ S \in   \mathcal{S}(\ell,s,M,\tau)
\mid M|_{E^{\pi}(S)}  \equiv M'
\}|\,.$$
$n_{M'}(\ell,s)$ is just the number of times that a term $Z_{M''}(G)$ arises
in the coefficient of $f_{\ell,s}(k)$
where $M''\equiv M'$.
 Now,
for $k>|D|$ we can rewrite Equation~\eqref{eq:anotherone} as
\begin{align}
        Z_M(\ourinst^{\pi,\tau}(k, G))\ \ &=
            \!\! \sum_{0\leq \ell < s\leq |D|} \!\!
                f_{\ell,s}(k)\,T^{\pi,\tau}_{M,\ell,s}(G)\,,\label{eq:lineq}
\intertext{where}
        T^{\pi,\tau}_{M,\ell,s}(G) \hspace{1em} &= \hspace{-1em}
                     \sum_{M'\in A_M^{\pi,\tau}(\ell,s)/{\equiv}}
                     \hspace{-1em}
                         n_{M'}(\ell,s)\, Z_{M'}(G)\,.\label{eq:inner}
    \end{align}

\begin{theorem}
\label{thm:interpolation}
Let $M$ be any symmetric matrix in $\{0,1,*\}^{D\times D}$.
Suppose that there are  $\ell$ and $s$
satisfying    $0\leq \ell < s \leq |D|$
and Boolean values $\pi$ and $\tau$ in $\{0,1\}$ such that,
     up to
    $\equiv$-equivalence, $A^{\pi,\tau}_M(\ell,s)$ contains either
    \begin{itemize}
    \item exactly one hard proper principal submatrix of~$M$; or
    \item exactly two hard proper principal submatrices of~$M$ and
        these are $\mtwo***0$ and $\mtwo***1$.
    \end{itemize}
    Then \nPartitions{M} is \numP{}-complete.
\end{theorem}
\begin{proof}
First, let's go back to Equation~\eqref{eq:lineq}.
Let $\numpairs= \binom{|D+1|}{2}$.
Note that $M$, $\pi$ and $\tau$ are all fixed.
Consider a graph~$G$.
In the proof we will consider the quantities
$T_{M,\ell,s}^{\pi,\tau}(G)$
to be a set of $\numpairs$ ``variables''
indexed by the pairs $(\ell,s)$.
We will compute the values of these variables
by making multiple
evaluations
of
$Z_M(J^{\pi,\tau}(k,G))$ for different values of~$k$
(using an oracle for \nPartitions{M}).

It will help to have an enumeration of the $\numpairs$ pairs $(\ell,s)$
with $0\leq \ell < s \leq |D|$, so let
$(\ell_j,s_j)$ be the $j$'th such pair (for $1 \leq j \leq \numpairs$).
Choose $\numpairs$ distinct values $k_1,\ldots,k_\numpairs$,
 which meet the requirements of Lemma~\ref{lemma:lin-ind}.
Let $F$ be the $\numpairs\times \numpairs$ integer matrix whose $(i,j)$'th entry
$F_{i,j}$ is $f_{\ell_j,s_j}(k_i)$.

Using an oracle for \nPartitions{M}, we can compute
the entries of a length-$\numpairs$ column vector $\overline{Z}$
whose $i$'th entry is
$Z_M(J^{\pi,\tau}(k_i,G))$.

Let $\overline{T}$ be a length-$\numpairs$ column vector whose
$j$'th entry is the $j$'th variable $T_{M,\ell_j,s_j}^{\pi,\tau}(G)$.
Then Equation~\eqref{eq:lineq}
gives the system of equations
$ \overline{Z} = F \overline{T}$.

Lemma~\ref{lemma:lin-ind} shows that
$F$ has full rank
so $F$ can be inverted, and we can
compute all of the variables $T_{M,\ell,s}^{\pi,\tau}(G)$
using $F^{-1} \overline{Z} = \overline{T}$
and using the \nPartitions{M} oracle to compute the values of $\overline{Z}$.

By Equation~\eqref{eq:inner}, each
variable
    $T^{\pi,\tau}_{M,\ell,s}(G)$ is a sum of terms, each of which
    is a constant multiple of $Z_M(G)$ or of $Z_{M'}(G)$ for some
    proper principal submatrix $M'$ of~$M$.  If exactly one of these
    submatrices~$M'$ is hard, we can use   the
    polynomial-time algorithms for the other problems $Z_{M''}$
    ($M''\not\equiv M'$) to compute $Z_{M'}(G)$ in polynomial time.
    If exactly two of the submatrices~$M'$ are hard and these are
    $\mtwo***0$ and $\mtwo***1$, we can similarly compute
    $\alpha\IS(G) + \beta\Clique(G)$ in polynomial time for constants
    $\alpha,\beta\geq 1$, which is \numP{}-complete by
    Lemma~\ref{lemma:IS-clique}.  In both cases, we conclude that
    \nPartitions{M} is \numP{}-complete.
\end{proof}

We could, in fact, go further and consider the
equations~\eqref{eq:inner} for different values for $\ell$ and~$s$ as
a system of linear equations in variables $Z_{M'}(G)$ for principal
submatrices $M'$ of~$M$.  This system may be underdetermined so it
might not be possible to solve for all the terms $Z_{M'}(G)$ that
appear; however, we do not necessarily need to.  We can still deduce
\numP{}-completeness for any matrix~$M$ for which we can solve the
equations for at least one variable $Z_{M'}(G)$ where $M'$ is a hard
proper principal submatrix.
Similarly, we can still deduce \numP{}-completeness for any matrix~$M$ for which
we can solve the equations for a linear combination of $Z_{M'}(G)$ and $Z_{M''}(G)$
where $M'$ and $M''$ are equivalent to $\mtwo***0$ and $\mtwo***1$.

 It turns out that this extension of our technique is not necessary for
$4\times 4$ matrices, apart from one exceptional case which we resolve
by hand; but this extension would be required to extend the technique
to larger matrices.

Theorem~\ref{thm:interpolation} allows us to show that our example
matrix is hard.  Recall that the matrix is
\begin{equation*}
    M\;=\;\bordermatrix{
              & a & b & c & d \cr
          a\; & 0 & 0 & 1 & * \cr
          b\; & 0 & 0 & 1 & 1 \cr
          c\; & 1 & 1 & 1 & 1 \cr
          d\; & * & 1 & 1 & *}
\end{equation*}
and consider again the graph $\ourinst^{0,0}(k, G)$ for some $k>4$ and
some~$G$.
For $S\in\{a,b,d,ab,ad\}$ we find that
$S\not\in \badones(M,0)$ so there are $S$-surjective $M$-partitions of $\Kkbar$.
Thus, we have

\begin{center}
\begin{tabular}{c|ccccc}
    $S$             & $a$     & $b$      & $d$      & $ab$     & $ad$ \\
\hline
    $\vphantom{\Big(}Z_M^S(\Kkbar)$ & $f_{0,1}$ & $f_{0,1}$ & $f_{0,1}$ & $f_{0,2}$ & $f_{0,2}$ \\
    $E^0(S)$        & $abd$   & $ab$     & $ad$     & $ab$     & $ad$ \\
    $M|_{E^0(S)}$      & hard    & easy     & hard     & easy     & hard
\end{tabular}
\end{center}

Equation~\eqref{eq:lineq} gives
\begin{align*}
    Z_M(\ourinst^{0,0}(k,G)) &= f_{0,1}(k)\,T^{\pi,\tau}_{M,0,1}(G) + f_{0,2}(k)\,T^{\pi,\tau}_{M,0,2}(G)\,,
\intertext{where}
    T^{\pi,\tau}_{M,0,1}(G) &= Z_{M|_{abd}}(G) + Z_{M|_{ab}}(G) + Z_{M|_{ad}}(G) \\
    T^{\pi,\tau}_{M,0,2}(G) &= Z_{M|_{ab}}(G) + Z_{M|_{ad}}(G)\,.
\end{align*}
$T^{\pi,\tau}_{M,0,1}(G)$ contains two terms that are partition functions of hard
matrices so is not useful to us but $T^{\pi,\tau}_{M,0,2}(G)$ contains only one
($Z_{M|_{ad}}$, which counts independent sets).  Therefore, by
Theorem~\ref{thm:interpolation}, \nPartitions{M} is \numP{}-complete.
Given an oracle for $Z_M$, we could obtain the value of $T^{\pi,\tau}_{M,0,2}(G)$ by
interpolation and, from that, we could compute $Z_{M|_{ad}}$.

\section{The computer-assisted dichotomy}
\label{sec:computer}

So far, we have seen three techniques for determining the
computational complexity of the \nPartitions{M} problem for a given
matrix~$M$.  If $M$~is pure, \nPartitions{M} is a graph homomorphism
problem, so $M$ is hard if, and only if, it has a $2\times 2$ submatrix
containing exactly three~$*$s.  For impure~$M$,
Lemma~\ref{lemma:doubletons} allows us to identify a class of
tractable matrices and the techniques of Section~\ref{sec:hard} allow
us to identify a class of hard matrices.
We were unable to prove that
the last two cases cover all impure $4\times 4$ matrices, so we wrote
a computer program to check all such matrices, as follows.

The number of distinct symmetric $4\times 4$ $\{0,1,*\}$-matrices is
modest: at most $3^{10} = 59,049$.  Thus, from a computational point
of view it is not necessary to do anything to reduce the search space.
However, it turns out that the methods described above are not enough
to determine the complexity of \nPartitions{M} for all symmetric
$4\times 4$ matrices.  Recall that $M_1\approx M_2$ if $M_1\equiv M_2$
or $\overline{M_1}\equiv M_2$ (i.e., $M_1$~can be transformed
into~$M_2$ by permuting~$D$ and possibly exchanging $0$s and~$1$s).
Since \nPartitions{M_1} and \nPartitions{M_2} are computationally
equivalent when $M_1\approx M_2$, it suffices to consider only one
matrix from each $\approx$-equivalence class.  This minimises the set
of matrices that the program fails to resolve.

To do this, we associate each $4\times 4$ symmetric matrix~$M$ with
the string
\begin{equation*}
 w(M) = M_{a,a} M_{b,b} M_{c,c} M_{d,d} M_{a,b} M_{b,c} M_{c,d} M_{a,c} M_{b,d} M_{a,d}
      \in \{0,1,*\}^{10}\,.
\end{equation*}
The program generates $4\times 4$ matrices in the lexicographic order
induced by taking $0<1<*$.  For each matrix~$M$, we check whether
$w(M') < w(M)$ for any matrix $M'\approx M$.  If there is such
an~$M'\!$, we have already considered a matrix equivalent to~$M$ so we
do not need to consider it again.

For each matrix $M$ that survives (i.e., for the lexicographically
first member of every $\approx$-equivalence class), we apply the
following tests. The correctness of these tests will be explained below.

\begin{enumerate}
\item \label{stepone} If $M$ is pure (contains no $0$'s or no $1$'s)
\begin{enumerate}
\item If $M$ contains a $2\times 2$ submatrix with exactly
three $*$s then \nPartitions{M} is \numP-complete.
\item Otherwise, \nPartitions{M} is in \FP.
\end{enumerate}

\item \label{steptwo} Otherwise, if the test of Lemma~\ref{lemma:doubletons} shows
    that $M$~has no derectangularising sequence then
    \nPartitions{M} is  in \FP{}.

\item \label{stepthree} Otherwise, for each proper principal submatrix $M'$ of~$M$, we
    can determine whether $M'$ is easy or hard using the
    characterisations of Hell, Hermann and Nevisi~\cite{HHN} and Dyer
    and Greenhill~\cite{DG}.
    The program now does the following
    for each $\pi,\tau\in\{0,1\}$,
    and each $0\leq \ell < s \leq |D|$,
    using
    the notation of Section~\ref{sec:hard}.
    It computes the elements of $A_M^{\pi,\tau}(\ell,s)$, up to $\equiv$-equivalence
    and makes the following conclusions.

    \begin{enumerate}
    \item
    If this set contains exactly one hard proper principal submatrix of $M$
    then \nPartitions{M} is \numP-complete.
    \item  If this set contains exactly two hard proper submatrices
    of~$M$ and these are $\mtwo***0$ and $\mtwo***1$
    then \nPartitions{M} is \numP-complete.
    \end{enumerate}

\item If none of the above tests resolves the complexity of
    \nPartitions{M}, output the matrix as having unknown complexity.
\end{enumerate}

The program resolves the complexity of \nPartitions{M} for all
but six $\approx$-equivalence classes of matrices.  These six are handled in
the next section; all turn out to be hard.

We conclude this section by justifying the correctness of the program.
  If $M$ is pure  then
    \nPartitions{M} is equivalent to a homomorphism-counting problem
 so the correctness of Step~\ref{stepone} follows from
 the dichotomy theorem of   Dyer
    and Greenhill~\cite{DG}.
 Now consider Step~\ref{steptwo}.    If  $M$~has no derectangularising sequence then
    \nListPartitions{M} is in \FP{} by Theorem~\ref{thm:list}.  Since
    \nPartitions{M} is just the special case where every vertex has
    list~$D$, \nPartitions{M} is also in \FP{}.
    Finally, the correctness of Step~\ref{stepthree} follows from Theorem~\ref{thm:interpolation}.

\section{The last six matrices}
\label{sec:handproofs}

In this section, we despatch the six matrices
that
our program could not resolve.

\subsection{Bipartite problems}

Let $G=(U,V,E)$ be a bipartite graph and let its \emph{bipartite
  complement} be the graph $(U,V,(U\times V)\setminus E)$.
  Note that the bipartite complement of~$G$ depends on the partition $(U,V)$
  and not just on the vertices and edges of~$G$.
  A
\emph{bipartite clique} in~$G$ is a set $S\subseteq U\cup V$ such that
$G$~contains an edge between every vertex of $S\cap U$ and every
vertex of $S\cap V$.  Note the trivial case that $S$~is a bipartite
clique in~$G$ if $S\subseteq U$ or $S\subseteq V$.

Counting
bipartite cliques in a bipartite graph is \numP{}-complete.  This is
because a bipartite clique in~$G$ is an independent set in $G$'s
bipartite complement and counting independent sets in a bipartite
graph is \numP{}-complete~\cite{ProvanBall}.  The problem of
counting bipartite cliques remains \numP{}-complete when the input is
restricted to be a connected bipartite graph.
To see this, note that counting non-trivial bipartite cliques (with at least one
edge) is
inter-reducible with the problem of counting all bipartite cliques
(since the number of trivial ones is easy to compute). But the number of non-trivial bipartite cliques in a graph   is the sum of the numbers of non-trivial bipartite cliques in each
component.
 
\begin{lemma}
    \nPartitions{M} is \numP{}-complete for
    \begin{equation*}
        M\;=\;\bordermatrix{
                & a & b & c & d \cr
            a\; & 0 & 0 & * & * \cr
            b\; & 0 & 0 & 1 & * \cr
            c\; & * & 1 & 0 & 0 \cr
            d\; & * & * & 0 & 0}.
    \end{equation*}
\end{lemma}
\begin{proof}

The problem of counting bipartite cliques in a connected bipartite graph
reduces immediately to counting $M$-partitions.
Consider a connected bipartite graph~$G$ with
vertex bipartition~$(U,V)$. Since $G$~is connected,
any $M$-partition of~$G$
either
\begin{itemize}
\item assigns vertices in~$U$ to parts~$a$ and~$b$
and assigns vertices in~$V$ to parts~$c$ and~$d$, or
\item assigns vertices in~$U$ to parts~$c$ and~$d$
and assigns vertices in~$V$ to parts~$a$ and~$b$.
\end{itemize}

In each case, the vertices in parts~$b$ and~$c$ form a bipartite
clique because $M_{b,c}=1$ whereas the other
relevant entries of~$M$ are all stars.

So the $M$-partitions in each case are in one-to-one correspondence
with the bipartite cliques of~$G$.
Therefore, $Z_M(G)$ is twice the number of bipartite cliques
in~$G$.
\end{proof}

\begin{lemma}
    \nPartitions{M} is \numP{}-complete for $M\in \{M_1, M_2, M_3\}$, where
    \begin{equation*} 
        M_1\;=\;\bordermatrix{
                & a & b & c & d  \cr
            a\; & 0 & 0 & * & *  \cr
            b\; & 0 & 0 & 0 & *  \cr
            c\; & * & 0 & 1 & 1  \cr
            d\; & * & * & 1 & 1 }
        \qquad 
        M_2 = \bordermatrix{
                & a & b & c & d \cr
            a\; & 0 & 0 & * & * \cr
            b\; & 0 & 0 & 0 & * \cr
            c\; & * & 0 & 1 & * \cr
            d\; & * & * & * & 1 }
        \qquad 
        M_3 = \bordermatrix{
                & a & b & c & d  \cr
            a\; & 0 & * & * & *  \cr
            b\; & * & 0 & 0 & *  \cr
            c\; & * & 0 & 1 & *  \cr
            d\; & * & * & * & 1  }\,.
    \end{equation*}
    \label{lem:threematrices}
\end{lemma}
\begin{proof}
    In all three cases $M\in \{M_1,M_2,M_3\}$ we will show how to
    reduce from the \numP{}-complete problem of counting independent
    sets in a bipartite graph to counting $M$-partitions.

    Let $G$ be a bipartite graph with vertex bipartition $(U,V)$.  For
    an integer~$k>4$, construct $G_k$ from~$G$ by adding a set~$W$ of
    $k$~new vertices and  adding all edges between distinct vertices $w$ and $v$
    where $w\in W$ and $v\in V \cup W$. Note that $G_k$ is not bipartite
    because the vertices of~$W$ form a complete subgraph.

This complete subgraph is
the same as the gadget $\Kkplus$ that we have already considered, so
it will be useful to
apply
Theorem~\ref{thm:clique-partition} to all $3$-element sets $S\subseteq D$.
The outcomes for $k>4$ are:
\begin{center}
\begin{tabular}{c|cccc}
    $S$             & $abc$ & $abd$ & $acd$ & $bcd$ \\
\hline
    $\vphantom{\Big(}Z_{M_1}^S(\Kkplus)$ & $0$ & $0$ & $f_{1,3}(k)$ & $0$ \\
    $Z_{M_2}^S(\Kkplus)$ & $0$ & $0$ & $f_{1,3}(k)$ & $0$ \\
    $Z_{M_3}^S(\Kkplus)$ & $0$ & $f_{2,3}(k)$ & $f_{1,3}(k)$ & $0$
\end{tabular}
\end{center}

By Theorem~\ref{thm:clique-partition},
no  set $S$ with $|S|\neq 3$ has $Z_M^S(\Kkplus) = f_{1,3}(k)$.
The exact value of $Z_M^S(\Kkplus)$ as a function of~$k$ for such a set~$S$
will not
be important in the following interpolation argument.

 Now
 consider $S\subseteq D$
 so that $Z_M^S(\Kkplus) = f_{\ell,|S|}(k)$.
 Let $Z_M^{W\mapsto S}(G)$ denote the
 number of $M$-partitions of~$G$
 in which every vertex of~$U$ is assigned a part
 in $E^0(S)$ and every vertex of~$V$ is assigned a part in
 $E^1(S)$.
 This is the number of $M$-partitions of~$G$
 that can be combined with an $S$-surjective
 $M$-partition of
the $k$-clique on~$W$  to get a valid $M$-partition of~$G_k$.
  We will use interpolation  as in the proof of
 Theorem~\ref{thm:interpolation}.
  Suppose $k>4$.
 Using the table above and noting the value~$f_{1,3}(k)$ in the
 $acd$~column, we can write
 $$Z_{M}(G_k) =
 f_{1,3}(k)\, Z_{M}^{W\mapsto \{a,c,d\}}(G)\ \  + \!\!\!
 \sum_{\substack{0 \leq \ell < s \leq 4,\\
 (\ell,s) \neq (1,3)}}\
 \sum_{\substack{S \subseteq D,\\ |S|=s}}
 \mathbf{1}_{Z_{M}^S(\Kkplus)=f_{\ell,s}(k)}\,
 f_{\ell,s}(k)\,
 Z_{M}^{W\mapsto S}(G)\,,$$
 where
 $ \mathbf{1}_{Z_{M}^S(\Kkplus)=f_{\ell,s}(k)}$
 is the indicator for the event
 that $Z_{M}^S(\Kkplus)=f_{\ell,s}(k)$ --
 we know from Theorem~\ref{thm:clique-partition} that,
if this event does not hold, then
 $Z_{M}^S(\Kkplus)=0$.

 As in the proof of Theorem~\ref{thm:interpolation},
 Lemma~\ref{lemma:lin-ind} guarantees that the
  $f_{\ell,s}(k)$ values are
    linearly independent.
So, by varying~$k$ and using an oracle for \nPartitions{M} to
compute the left-hand side, we can compute the
coefficient of $f_{1,3}(k)$, which is
$Z_{M}^{W\mapsto \{a,c,d\}}(G)$.
So, to finish the proof, we just
need to show that computing
$Z_{M}^{W\mapsto \{a,c,d\}}(G)$ is \numP{}-hard.

$Z_{M}^{W\mapsto \{a,c,d\}}(G)$
is the number of $M$-partitions of~$G$
in which every vertex of~$U$ is assigned to a part in   $E^0(\{a,c,d\})=\{a,b\}$
and every vertex of~$V$ is assigned to a part in  $E^1(\{a,c,d\}) = \{c,d\}$.
But note that
edges are forbidden between part~$b$ and part~$c$
so there is a one-to-one correspondence between
these partitions of~$G$
  and the independent sets of~$G$.
  (Vertices assigned to these parts are
  in the corresponding independent set.)
The result follows, since computing independent sets
of a bipartite graph~$G$ is \numP{}-hard.
 \end{proof}

The proof of the following lemma is similar in spirit but with more
details to track.
 
\begin{lemma}
\label{lem:hand-iii}
    \nPartitions{M} is \numP{}-complete for
    \begin{equation*}
        M\;=\;\bordermatrix{
                & a & b & c & d  \cr
            a\; & 0 & 0 & * & *  \cr
            b\; & 0 & 0 & 1 & *  \cr
            c\; & * & 1 & 1 & *  \cr
            d\; & * & * & * & 1  }\,.
    \end{equation*}
\end{lemma}

\begin{proof}
    Let $G$ be a bipartite graph with vertex bipartition $(U,V)$.  For
    any integer~$k>4$, construct $G_k$ from $G$ as follows.  The
    vertices of $G_k$ are $U\cup V\cup W\cup \{x_c, x_d\}$, where
    $|W|=k$ and $W$, $x_c$ and $x_d$ are new vertices.  The edges are
    as follows (see Figure~\ref{fig:Gk-construction}):
    \begin{itemize}
    \item every edge $(x,y)$ for $x\in \{x_c,x_d\}$, $y\in V\cup W$;
    \item every edge $(v,w)$ for $v\in V$, $w\in W$;
    \item every edge $(v,v')$ for distinct $v,v'\in V$;
    \item every edge $(x_c,u)$ for $u\in U$;
    \item every edge $(u,v)$ where $u\in U$, $v\in V$ and $(u,v)\notin
        E(G)$.
    \end{itemize}

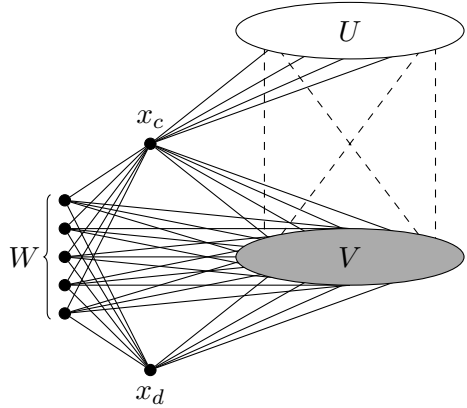
\begin{figure}[t]
\begin{center}
\begin{tikzpicture}[scale=0.75]
    \tikzstyle{vertex}=[fill=black, draw=black, circle, inner sep=1.5pt]

    \node[vertex] (xc) at (1.5,2) [label=90:{$x_c$}] {};
    \node[vertex] (xd) at (1.5,-2) [label=-90:{$x_d$}] {};
    \foreach \x in {1,...,5}
        \node[vertex] (w\x) at (0,-1.5+0.5*\x) {};  

    \draw[decorate,decoration=brace] (-0.25,-1.1) -- (-0.25,1.1);
    \draw (-0.25,0) node[anchor=east] {$W$};

    \foreach \x in {4,...,7} {
        \draw (xc) -- (\x,4);   
        \draw (xc) -- (\x,0);   
        \draw (xd) -- (\x,0);   
    }
    \foreach \x in {1,...,5}
        \draw (xc) -- (w\x) -- (xd);   

    \foreach \x in {1,...,5}
        \foreach \y in {0.5,0,-0.5}
            \draw (w\x) -- (5,\y);

    \draw[dashed] (3.5,4) -- (3.5,0) -- (6.5,4) -- (6.5,0) -- (3.5,4);

    \fill[color=white] (5,4) circle (2cm and 0.5cm); 
    \fill[color=black!33!white] (5,0) circle (2cm and 0.5cm); 

    \draw (5,4) circle (2cm and 0.5cm) node {$U$};
    \draw (5,0) circle (2cm and 0.5cm) node {$V$};
\end{tikzpicture}
\caption{The construction of $G_k$, used in the proof of
    Lemma~\ref{lem:hand-iii}, shown with $k=5$.  The dotted lines
    $U$--$V$ denote the complement of $G$'s edge relation between $U$
    and~$V$; the shading of $V$ indicates a clique on those vertices.}
\label{fig:Gk-construction}
\end{center}
\end{figure}

The subgraph induced on~$W$ is an independent set, so it is the same as $\Kkbar$.
We now apply
Theorem~\ref{thm:clique-partition} to all $2$-element sets $S\subseteq D$.
The outcomes for $k>4$ are:
\begin{center}
\begin{tabular}{c|cccccc}
    $S$             & $ab$ & $ac$ & $ad$ & $bc$ & $bd$ & $cd$ \\
\hline
    $\vphantom{\Big(}Z_{M}^S(\Kkbar)$ & $f_{0,2}(k)$ & $f_{1,2}(k)$ &
    $f_{1,2}(k)$ & $0$ & $f_{1,2}(k)$ & $0$\\
\end{tabular}
\end{center}
 By Theorem~\ref{thm:clique-partition},
no  set $S$ with $|S|\neq 2$ has $Z_M^S(\Kkbar) = f_{0,2}(k)$ as a
function of~$k$.

Now
 consider $S\subseteq D$
 so that $Z_M^S(\Kkbar) = f_{\ell,|S|}(k)$ for some $\ell<|S|$.
 Let $G_k-W$ denote the
 subgraph of~$G_k$ induced by all vertices other than those in~$W$.
 Let $Z_M^{W\mapsto S}(G)$ denote the
 number of $M$-partitions of~$G_k-W$
 in which every vertex of~$U$ is assigned a part
 in $E^0(S)$ and every other vertex is assigned a part in
 $E^1(S)$.
 As in the proof of Lemma~\ref{lem:threematrices}, each such
 $M$-partition of $G_k-W$ extends to $Z_M^S(\Kkbar)$ $M$-partitions
 of~$G_k$, so we can write
$$Z_{M}(G_k) =
 f_{0,2}(k)\, Z_{M}^{W\mapsto \{a,b\}}(G) \ \ + \!\!\!
 \sum_{\substack{0 \leq \ell < s \leq 4,\\
 (\ell,s) \neq (0,2)}}\
 \sum_{\substack{S \subseteq D,\\|S|=s}}
 \mathbf{1}_{Z_{M}^S(\Kkbar)=f_{\ell,s}(k)}\,
 f_{\ell,s}(k)\,
 Z_{M}^{W\mapsto S}(G)\,.$$

 As in the proof of Theorem~\ref{thm:interpolation},
 Lemma~\ref{lemma:lin-ind} guarantees that the
  $f_{\ell,s}(k)$ values are
    linearly independent.
So, by varying~$k$ and using an oracle for \nPartitions{M} to
compute the left-hand side, we can compute the
coefficient of $f_{0,2}(k)$, which is
$Z_{M}^{W\mapsto \{a,b\}}(G)$.
So to finish the proof, we just
need to show that computing
$Z_{M}^{W\mapsto \{a,b\}}(G)$ is \numP{}-hard.

 $Z_M^{W\mapsto \{a,b\}}(G)$ is the
 number of $M$-partitions of~$G_k-W$
 in which every vertex of~$U$ is assigned to a part
 in $E^0(\{a,b\})=\{a,b,d\}$ and every
 vertex in $V\cup \{x_c,x_d\}$
 is assigned to a part in
 $E^1(\{a,b\})=\{c,d\}$.

Note that the edges between vertices in~$V$ add no further restriction
on the parts assigned to vertices in~$V$, since $M|_{cd}$ contains no zeroes.
Vertices $x_c$ and $x_d$  are not adjacent, so one of them is assigned to part~$c$ and the other to
part~$d$.

In the first case, $x_c$ is assigned to part~$c$ and $x_d$~is assigned to part~$d$.
Vertices in~$U$ are adjacent to part~$c$ and not to part~$d$.
Since they are not adjacent to part~$d$, and we already know (from above)
that they are not assigned to part~$c$, each must be assigned to part~$a$ or~$b$.
So we have selected $M$-partitions in which vertices in~$U$ are
assigned to parts~$a$ or~$b$ and vertices in~$V$ are assigned
to parts~$c$ or~$d$.
This counts independent sets in~$G$, with
parts~$b$ and~$c$ corresponding to being in the independent set
(all edges between these parts must exist in~$G_k$,
which corresponds to an independent set in~$G$).

In the second case, $x_c$ is assigned to part~$d$ and $x_d$~is in part~$c$.
Vertices in~$U$ are adjacent to part~$d$ and not to part~$c$
so they can only be in parts $a$ and~$d$.
Since $U$ is an independent set and $M_{d,d}=1$, at most one of its vertices is in part~$d$.
We can count all such $M$-partitions in polynomial time by considering each possible vertex
$u\in U$ that might be assigned to part~$d$ and assigning the rest to part~$a$.
The vertex in part~$d$ restricts its non-neighbours in~$V$ to be assigned part~$c$.
The rest of the vertices in~$V$ can be assigned to either~$c$ or~$d$.

In conclusion, computing $Z_M(\cdot)$ enables us to compute
$Z_M^{W\mapsto \{a,b\}}(G)$. But computing $Z_M^{W\mapsto \{a,b\}}(G)$ enables us
to count independent sets of~$G$. Since
counting independent sets of a bipartite graph is \numP-hard, so is
counting $M$-partitions.
\end{proof}

\subsection{A matrix proved hard by solving simultaneous linear equations}
\label{sec:handproofs:linear}

Recall the definition of $T^{\pi,\tau}_{M,\ell,s}(G)$ from \eqref{eq:inner}.
In Section~\ref{sec:hard}, our gadgets were large cliques and
independent sets and we used interpolation on the number of vertices
in the gadget to compute $Z_{M'}(G)$ for some submatrix $M'$ such that
\nPartitions{M'} is \numP{}-complete.  Our final case is a matrix~$M$
where this technique only allows us to compute the linear combinations
$T^{\pi,\tau}_{M,\ell,s}(G) = \sum_i \alpha_i Z_{M_i}(G)$ where, although each subproblem
\nPartitions{M_i} is hard, we do not have enough independent linear
equations to compute any single term $Z_{M_i}(G)$.  The solution is to
use a similar gadget to generate an extra linear equation that allows
us to solve for a hard~$Z_{M'}$.

\begin{lemma}
\label{lemma:hand-iv}
    \nPartitions{M} is \numP{}-complete for
    \begin{equation*}
        M\;=\;\bordermatrix{
                & a & b & c & d  \cr
            a\; & 0 & * & * & *  \cr
            b\; & * & * & 0 & *  \cr
            c\; & * & 0 & * & 1  \cr
            d\; & * & * & 1 & *  }\,. 
    \end{equation*}
\end{lemma}
\begin{proof}
    We show how to reduce \nPartitions{M|_{abd}} to \nPartitions{M}.
    The matrix~$M|_{abd}$ is hard by Hell, Hermann and Nevisi's
    characterisation of the hard $3\times 3$ matrices~\cite{HHN}: the
    principal submatrix $M|_{ad} = \mtwo0***$ is hard.

    First, consider $J^{1,0}(k,G)$
    and the set of $M$-partitions in which
    the vertices of the $\Kkbar$ appear in exactly two parts,
    corresponding to the term~$T^{1,0}_{M,0,2}(G)$.  Note that every pair of parts
    is possible, except for $\{c,d\}$:
    \begin{center}
    \begin{tabular}{l|ccccc}
        $S$         & $ab$ & $ac$ & $ad$  & $bc$ & $bd$  \\
        \hline
        $\vphantom{\Big(}E^1(S)$    & $bd$ & $cd$ & $bcd$ & $ad$ & $abd$ \\
        $M|_{E^1(S)}$ & easy & easy & easy  & hard & hard
    \end{tabular}
    \end{center}
    Thus, noting that $M|_{ab} = M|_{ad} = \mtwo0***$, there is a
    polynomial-time-computable function $p(G)$ (corresponding to the
    ``easy'' entries of the above table) such that
    \begin{equation}
    \label{eq:normalgadget}
        T^{1,0}_{M,0,2}(G) = p(G) + Z_{M|_{abd}}(G) + Z_{M|_{ad}}(G)\,.
    \end{equation}

    Second, consider $J^{1,0}(k, G+x)$ where $G+x$ denotes the union
    of $G$ and a new isolated vertex~$x$.  We again consider
    $M$-partitions of this graph in which vertices of the~$\Kkbar$
    appear in exactly two parts and we divide up these partitions
    according to the part in which the new vertex~$x$ appears.  If the
    vertices of the~$\Kkbar$ are in parts $S\subset D$ and $x$~is in
    part~$i$, then the vertices of~$G$ must be in some subset of the
    parts $P(i,S) = E^0(\{i\})\cap E^1(S)$.  The possible
    combinations are as follows.

    \begin{center}
    \small
    \begin{tabular}{l|cc|ccc|cc|ccccc}
        $i$ & \multicolumn{2}{|c|}{$a$} & \multicolumn{3}{|c|}{$b$} & \multicolumn{2}{|c|}{$c$} & \multicolumn{5}{|c}{$d$} \\
        $S$ &       $bc$   & $bd$  & $ab$ & $ad$  & $bd$  & $ac$ & $ad$ & $ab$ & $ac$ & $ad$ & $bc$ & $bd$  \\
        \hline
        $P(i,S)$    & $ad$ & $abd$ & $bd$ & $bcd$ & $abd$ & $c$  & $bc$ & $bd$ & $d$  & $bd$ & $ad$ & $abd$ \\
        $M|_{P(i,S)}$ & hard & hard  & easy & easy  & hard  & easy & easy & easy & easy & easy & hard & hard
    \end{tabular}
    \end{center}

    This gives a second equation,
    \begin{equation}
    \label{eq:specialgadget}
        T^{1,0}_{M,0,2}(G+x) = p'(G) + 3Z_{M|_{abd}}(G) + 2Z_{M|_{ad}}(G)\,,
    \end{equation}
    where, again, $p'(G)$ is a polynomial-time computable function.

    As in the proof of Theorem~\ref{thm:interpolation}, we can compute
    $T^{1,0}_{M,0,2}(G)$ and~$T^{1,0}_{M,0,2}(G+x)$ by interpolations on~$k$, using an oracle
    for \nPartitions{M}.  Thus, we can solve \eqref{eq:normalgadget}
    and~\eqref{eq:specialgadget} for $Z_{M|_{abd}}(G)$ (and
    $Z_{M|_{ad}}(G)$, which is also \numP{}-hard), completing the
    reduction.
\end{proof}

\section{The dichotomy for $4\times 4$ matrices}
\label{sec:conjecture}

Finally, we establish Theorem~\ref{thm:dichotomy} and show that
Conjecture~\ref{conjecturetoquote} holds for $4\times 4$ matrices.

\begin{theorem}
    Let $M$ be a symmetric matrix in $\{0,1,*\}^{4\times 4}\!$.  Then
    \nPartitions{M} is \numP{}-complete if $M$~has a
    derectangularising sequence, and is in \FP{}, otherwise.
\end{theorem}
\begin{proof}
The conjecture is already know to hold for pure matrices (see Section~\ref{sec:small}).

    The impure matrices covered by Lemma~\ref{lemma:doubletons} have
    no derectangularising sequence, so are easy by
    Theorem~\ref{thm:list}.

    For the matrices proved hard via Theorem~\ref{thm:interpolation},
    the computer program finds a hard principal submatrix of size
    either $2\times 2$ or $3\times 3$.  If the $2\times 2$ submatrix
    $M|_S$ is hard, then $S,S$ is a derectangularising sequence; if
    the $3\times 3$ submatrix $M|_T$ is hard then, by
    Lemma~\ref{lemma:HHNrephrased}, $M|_T$ has a derectangularising
    sequence, and this is also derectangularising for~$M$.

    For each of the six matrices proved hard in
    Section~\ref{sec:handproofs},
    it is easy to check that $\{a,b\},\{c,d\}$  is a
    derectangularising sequence.
\end{proof}

\section{Acknowledgements}

We thank the referees for useful suggestions.
We thank John Lapinskas for spotting
an error in an earlier proof of
 Lemma~\ref{lemma:lin-ind} and the journal staff, who were willing to make the changes
 at a late stage.

\bibliographystyle{plain}

\bibliography{\jobname}{}

\end{document}